\documentclass[]{interact}

\usepackage{epstopdf}%
\usepackage[caption=false]{subfig}%

\usepackage[natbibapa,nodoi]{apacite}
\usepackage{amsmath,amssymb,amsfonts}
\usepackage{graphicx}
\usepackage{mathtools}
\usepackage{enumerate}
\usepackage{enumitem}
\usepackage{xcolor}
\usepackage{changepage}

\usepackage{caption} 
\theoremstyle{plain}%
\newtheorem{theorem}{Theorem}[section]
\newtheorem{lemma}[theorem]{Lemma}

\theoremstyle{definition}

\theoremstyle{remark}
\newtheorem{remark}{Remark}

\usepackage{colortbl}
\newcommand{\kraxn}[2]{\lceil#1\rfloor^{#2}}

\newcommand{\vect}[1]{\boldsymbol{#1}}

\newcommand{\T}{^{\mathrm{T}}}

\newcommand{\pert}{\Delta}

\newcommand{\krax}[2]{\lfloor#1\rceil^{#2}}
\def\kraxn#1{\def\tempb{#1}\futurelet\next\kraxn@i}%
\def\kraxn@i{\ifx\next\bgroup\expandafter\kraxn@ii\else\expandafter\kraxn@end\fi}%
\def\kraxn@ii#1{{\sqrt{|\tempb|}\text{sign}\left(\tempb\right)}}%
\def\kraxn@end{{\text{sign}\left(\tempb\right)}}%

\newcommand{\sign}[1]{\text{sign}(#1)}
\newcommand{\ddt}[1]{\frac{\text{d}#1}{\text{d}t}}

\def\dert#1{\def\tempa{#1}\futurelet\next\dert@i}%
\def\dert@i{\ifx\next\bgroup\expandafter\dert@ii\else\expandafter\dert@end\fi}%
\def\dert@ii#1{{\dot{\tempa}_{#1}}}%
\def\dert@end{{\dot{\tempa}}}%

\DeclarePairedDelimiter\abs{\lvert}{\rvert}%
\DeclarePairedDelimiter\norm{\lVert}{\rVert}%

\makeatletter
\let\oldabs\abs
\def\abs{\@ifstar{\oldabs}{\oldabs*}}
\let\oldnorm\norm
\def\norm{\@ifstar{\oldnorm}{\oldnorm*}}
\makeatother

\newcommand{\rank}{\text{rank }}

\newcommand{\obsv}{\boldsymbol{\mathcal{O}}}
\newcommand{\obsvred}{\obsv_\text{R}}

\newcommand{\Abeta}{\vect{H}}
\newcommand{\bbeta}{\vect{w}}
\newcommand{\comma}{\text{,}}
\newcommand\undermat[2]{%
	\makebox[0pt][l]{$\smash{\underbrace{\phantom{%
					\begin{matrix}#2\end{matrix}}}_{\text{$#1$}}}$}#2}

\begin{document}

\title{Higher-order sliding mode observer design for linear time-invariant multivariable systems based on a new observer normal form}

\author{
\name{Helmut Niederwieser\textsuperscript{a,b,}*\thanks{*Corresponding author. Email: helmut.niederwieser@tugraz.at}, Markus Tranninger\textsuperscript{a}, Richard Seeber\textsuperscript{c} and Markus Reichhartinger\textsuperscript{a}}
\affil{\textsuperscript{a}Institute of Automation and Control, Graz University of Technology, Graz, Austria;\\ \textsuperscript{b}BEST - Bioenergy and Sustainable Technologies GmbH, Graz, Austria;\\
\textsuperscript{c}Christian Doppler Laboratory for Model Based Control of Complex Test Bed Systems, Institute of Automation and Control, Graz University of Technology, Graz, Austria;}
}

\maketitle

\begin{abstract}
In various applications in the field of control engineering the estimation of the state variables of dynamic systems in the presence of unknown inputs plays an important role. Existing methods require the so-called observer matching condition to be satisfied, rely on the boundedness of the state variables or exhibit an increased observer order of twice the plant order. In this article, a novel observer normal form for strongly observable linear time-invariant multivariable systems is proposed. In contrast to classical normal forms, the proposed approach also takes the unknown inputs into account. The proposed observer normal form allows for the straightforward construction of a higher-order sliding mode observer, which ensures global convergence of the estimation error within finite time even in the presence of unknown bounded inputs. Its application is not restricted to systems which satisfy the aforementioned limitations of already existing unknown input observers. The proposed approach can be exploited for the reconstruction of unknown inputs with bounded derivative and robust state-feedback control, which is shown by means of a tutorial example. Numerical simulations confirm the effectiveness of the presented work.
\end{abstract}


\section{Introduction}\label{sec:introduction}

The robust estimation of the state variables of dynamic systems from its outputs in the presence of unmeasured external disturbances and model uncertainties is of great interest in various applications in the field of control engineering, e.g. when dealing with fault detection problems \citep{Patton1989} or robust state-feedback control \citep{Yang2017}. One approach is to regard such external disturbances and model uncertainties as unknown inputs and to subsequently apply robust estimation methods, which are referred to as unknown input observers.

Kalman filter based state estimators, such as for example the classical Kalman filter \citep{Kalman1960b,Kalman1961} in the case of linear systems and the extended Kalman filter \citep{Sorenson1985} as well as the unscented Kalman filter \citep{Julier1997,Wan2000} in the case of nonlinear systems, allow to model unknown inputs as a Gaussian process noise with zero mean.
High-gain observers \citep{Khalil2014} suppress the influence of unknown inputs by means of a linear output injection with high gain.
In contrast to the aforementioned linear methods, sliding mode based observers are able to cope with a wider range of input uncertainties \citep{Shtessel2013,Utkin1999,Koshkouei2004,Fridman2006,De2013,Li2018}, e.g. when dealing with unknown inputs with bounded amplitude or bounded derivatives.

The applicability of conventional first-order sliding mode based approaches \citep{utkin2009sliding,edwards1998sliding,walcott1987state} as well as linear Luenberger like observers \citep{Darouach1994,Kudva1980,Yang1988,HAUTUS1983353} is limited to strongly detectable systems which additionally satisfy a rank condition often referred to as the observer matching condition.
In contrast, so-called step-by-step sliding mode observers and higher order sliding mode approaches do not require the observer matching condition to be fulfilled. Step-by-step sliding mode observers transform the system in a triangular structure and successively apply either classical first-order \citep{Barbot1996,Floquet2004} or second-order sliding mode methods \citep{Floquet2007,Floquet2007a} for reconstruction purposes.
However, in order to achieve global convergence of those observers, the boundedness of certain derivatives of the outputs is required which in general depend on the state variables. Hence, only local asymptotic stability of the estimation error dynamics is obtained in general.

Higher order sliding mode observers typically rely on the robust exact differentiator (RED) \citep{Levant1998,Levant2003}. Whenever it is possible to express the state variables by means of the outputs, their derivatives and the control inputs only, the RED can be directly applied as a state estimator. However, in general, also this procedure only permits local asymptotic stability of the estimation error dynamics.
In order to overcome this issue, a cascaded observer consisting of a non-robust observer and an RED is applied \citep{Fridman2011,Fridman2006,Shtessel2013,Tranninger2018,Tranninger2019}. This cascaded observer achieves global convergence in finite time in the presence unknown bounded inputs. 
However, the observer order is substantially larger than the order of the considered plant. Moreover, by the use of the additional non-robust observer, the number of observer parameters is increased. Since the parameters of the RED depend on the parameters of the additional observer as well as the characteristics of the unknown inputs, the tuning of those observer parameters may be cumbersome.

Recently, an RED-based unknown input observer for strongly observable linear time-invariant (LTI) systems with one single unknown input and one single output has been proposed \citep{Niederwieser2019}. It is based on the classical observability canonical form of linear systems.
This approach requires neither the application of an cascaded observer scheme nor the boundedness of the state variables of the plant for global convergence.
However, an extension to systems with multiple unknown inputs and multiple outputs is not available yet. In contrast to the single-input single-output case, existing normal forms are not suitable for the construction of such an observer. The well-known classical observability canonical form by Luenberger~\citep{Luenberger1967, Gupta1974}, which is based on observability indices, does not take the impact of the unknown inputs into account. The so-called special coordinate basis \citep{Chen2004, Sannuti1987} decomposes the system into coupled chains of integrators. However, the application of the RED as a state observer would again require the boundedness of the state variables for global convergence.

The contribution of this article is the development of a novel observer normal form for strongly observable LTI systems with multiple unknown inputs and multiple outputs. In this representation the system exhibits a lower block triangular structur up to an output injection which allows for a straightforward construction of an RED-based unknown input observer. The proposed observer does not rely on the boundedness of the state variables and does not additionally increase the observer order beyond the order of the plant. Moreover, the observer matching condition, which is prerequisite for the existence of many classical unknown input observers, is not required.

A detailed description of the considered problem statement is given in Section \ref{sec:problem_statement}. Section \ref{sec:proposed_canonical_form} presents the novel observer normal form developed for strongly observable LTI systems and a robust observer for systems in the proposed observer normal form. A description of the transformation into the proposed observer normal form is provided in Section \ref{sec:transformation_into_canonical_form}. The applicability of the presented approach is demonstrated by means of a simulation example in Section \ref{sec:Simulation_examples}. In addition, based on the provided example, it is shown how to exploit the presented approach for reconstructing the unknown input. Finally, Section \ref{sec:Conclusion} concludes the presented work.

\subsection{Preliminaries}\label{ssec:notation}

In order to abbreviate sign preserving power functions, the notation
\begin{align*}
\krax{\cdot}{\gamma}=|\cdot|^{\gamma}\sign{\cdot}\quad\text{and particularly} \quad	\krax{\cdot}{0}=\sign{\cdot}
\end{align*}
is used in this article. The notation
\begin{align*}
a \leftarrow f(a)
\end{align*}
refers to an update step for the value of $a$ required in the description of iterative algorithms, i.e., the new value of $a$ is some function $f(\cdot)$ of the old value of $a$.

The time-derivative $\ddt{}x$ is also denoted by $\dot{x}$. Moreover, the solutions of differential equations with discontinuous right-hand side are understood in the sense of Filippov \citep{Filippov2013a}.
The $i^{\text{th}}$ canonical unit vector is denoted by $\vect{e}_i$.
Furthermore, $\vect{I}_i$ refers to the $i \times i$ identity matrix and $\vect{0}_{j \times k}$ is the zero matrix of dimension $j \times k$.

\section{Problem statement}\label{sec:problem_statement}

Consider the strongly observable LTI system
\begin{align}\label{eq:original_system}
\begin{split}
\dot{\vect{x}} &=  \vect{A}\vect{x} + \vect{D}\vect{\pert}\text{,}\\
\vect{y} &= \vect{C}\vect{x}\comma
\end{split}
\end{align}
where ${\vect{x} = \begin{bmatrix} x_1 & x_2 & \dots & x_n \end{bmatrix}\T}$ is the $n$-dimensional state vector, ${\vect{\pert} = \begin{bmatrix} \pert_1 & \pert_2 & \dots & \pert_m \end{bmatrix}\T}$ denotes the vector of possibly time-varying, unknown inputs and ${\vect{y} = \begin{bmatrix} y_1 & y_2 & \dots & y_p \end{bmatrix}\T}$ is the output vector.
The amplitudes of the unknown inputs are assumed to be bounded, i.e.,
\begin{align}\label{eq:perturbation_bounds}
\sup_t \abs{\pert_i} \leq L_i \qquad \text{with} \qquad L_i > 0\comma \quad i = 1, \dots, m.
\end{align}
The unknown-input matrix $\vect{D} \in \mathbb{R}^{n \times m}$ and the output matrix $\vect{C} \in \mathbb{R}^{p \times n}$ are partitioned into column and row vectors, respectively, as
\begin{align}\label{eq:original_perturbation_and_output_matrix}
\vect{D} =
\begin{bmatrix}
\vect{d}_1 &
\vect{d}_2 &
\dots &
\vect{d}_m
\end{bmatrix}
\comma \qquad
\vect{C} =
\begin{bmatrix}
\vect{c}_1\T \\
\vect{c}_2\T \\
\vdots \\
\vect{c}_p\T \\
\end{bmatrix}.
\end{align}
Without loss of generality, the columns of $\vect{D}$ as well as the rows of $\vect{C}$ are assumed to be linearly independent, i.e., $\rank\vect{D} = m$ and $\rank\vect{C} = p$.
As stated above, the considered system is assumed to be strongly observable \citep{HAUTUS1983353}, i.e., the Rosenbrock matrix
\begin{align}\label{eq:Rosenbrock_matrix}
\vect{P}(s) =
\begin{bmatrix}
s \vect{I}_n - \vect{A} & -\vect{D} \\
\vect{C} & \vect{0}
\end{bmatrix}
\end{align}
satisfies
\begin{align}\label{eq:strong_observability_condition_rosenbrock}
\rank\vect{P}(s) = n + m \qquad \forall s \in \mathbb{C}.
\end{align}
The goal is to provide an exact estimate of the state vector $\vect{x}$ of system \eqref{eq:original_system} from the measured outputs $\vect{y}$ despite the unknown inputs $\vect{\pert}$. Note that existing state estimation methods for the considered system class have the following drawbacks:
\begin{itemize}
	\item Step-by-step sliding mode observers \citep{Barbot1996,Floquet2004,Floquet2007,Floquet2007a} as well as the direct application of the RED \citep{Levant1998,Levant2003} as a state observer require the boundedness of the plant's state variables for global convergence.
	\item Combined observers \citep{Fridman2006,Shtessel2013,Tranninger2018,Tranninger2019} consist of an RED and a Luenberger observer and ensure global convergence also in the case of unbounded state variables. However, the observer order and the number of tuning parameters are twice the system order.
	\item Classical Luenberger like observers \citep{Darouach1994,Kudva1980,Yang1988,HAUTUS1983353} as well as conventional first-order sliding mode based methods \citep{utkin2009sliding,edwards1998sliding,walcott1987state} require the so-called observer matching condition
	\begin{align}\label{eq:observer_matching_condition}
	\rank\vect{C}\vect{D} = \rank \vect{D}
	\end{align}
	to be satisfied.
\end{itemize} 
This article aims for the development of a novel unknown input observer which neither relies on bounded state variables, nor additionally increases the observer order beyond the order of the plant, nor requires the observer matching condition \eqref{eq:observer_matching_condition} to be satisfied.
For this purpose, a new observer normal form is proposed in Section \ref{sec:proposed_canonical_form} where the system is represented by means of $p$ coupled single-output systems which allow for a straightforward design of a robust observer. The corresponding transformation is derived in Section \ref{sec:transformation_into_canonical_form}.
\begin{remark}\label{rem:system_with_direct_feedthrough}
	In the case of a strongly observable system with direct feed-through, i.e.,
	\begin{align}\label{eq:system_with_direct_feed_through}
	\begin{split}
	\dot{\vect{x}} &= \vect{A}\vect{x} + \vect{D}\vect{\pert}\text{,}\\
	\vect{y} &= \vect{C}\vect{x} + \vect{F}\vect{\pert}\comma
	\end{split}
	\end{align}
	the state estimation problem can be reformulated in terms of a system without direct feed-through \citep[Section 5.4]{Chen2004}.
	This is achieved by output and input transformations
	\begin{align}\label{eq:system_with_direct_feed_through_transformations}
	\begin{split}
	\vect{\tilde{y}} = \begin{bmatrix}	\vect{\tilde{y}}_0 \\ \vect{\tilde{y}}_1 \end{bmatrix} 	= \vect{U}\vect{y}\comma \qquad
	\vect{\tilde{\pert}} = \begin{bmatrix}	\vect{\tilde{\pert}}_0 \\ \vect{\tilde{\pert}}_1 \end{bmatrix} 	= \vect{V}^{-1}\vect{\pert}\comma
	\end{split}
	\end{align}
	which allow to represent system \eqref{eq:system_with_direct_feed_through} in the form
	\begin{align}\label{eq:system_with_direct_feed_through_reduced}
	\begin{split}
	&\dot{\vect{x}} = \vect{\tilde{A}}\vect{x} + \vect{\tilde{D}}_1\vect{\tilde{\pert}}_1 + \vect{\tilde{D}}_0\vect{\tilde{y}}_0\text{,}\\
	&\vect{\tilde{y}}_1 = \vect{\tilde{C}}_1\vect{x}.
	\end{split}
	\end{align}
	A detailed description of this system transformation is given in Appendix \ref{app:remove_direct_feedthrough}. Since $\vect{\tilde{y}}_0$ is known it can be easily considered in the observer design. Hence, this term does not affect the construction of the observer. Furthermore, system \eqref{eq:system_with_direct_feed_through_reduced} is strongly observable, if the original system \eqref{eq:system_with_direct_feed_through} is strongly observable, which is shown in Appendix \ref{app:invariance_wrt_strong_observability}.
	Thus, the observer design algorithm proposed in this article is also applicable to systems with non-vanishing direct feed-through.
\end{remark}

\section{Observer normal form and robust observer design}\label{sec:proposed_canonical_form}

For the purpose of the robust observer design, a state transformation, which is described in detail in Section \ref{sec:transformation_into_canonical_form}, transforms the system \eqref{eq:original_system} into a novel observer normal form for multivariable LTI systems with unknown inputs. The proposed observer normal form introduced in Section \ref{ssec:proposed_canonical_form} decomposes the system into $p$ coupled single-output subsystems. Based on this observer normal form, a robust observer is designed in Section \ref{ssec:robust_observer_design} which allows for an exact finite-time reconstruction of the system's state variables.

\subsection{Proposed observer normal form}\label{ssec:proposed_canonical_form}

The main result of this article given by the proposed observer normal form is summarized in
\begin{theorem}\label{theorem:existence}
Let the LTI system \eqref{eq:original_system} be strongly observable, i.e., suppose its Rosenbrock matrix satisfies condition \eqref{eq:strong_observability_condition_rosenbrock}.
Then, there exist non-singular transformation matrices ${\vect{T}\in \mathbb{R}^{n\times n}}$ and ${\vect{\Gamma} \in \mathbb{R}^{p \times p}}$ such that the state transformation ${\vect{\bar{x}} = \vect{T}^{-1} \vect{x}}$ and the output transformation ${\vect{\bar{y}} = \vect{\Gamma} \vect{y}}$
yield the system in observer normal form
\begin{align}\label{eq:transformed_system}
    \begin{split}
        \overset{\vect{.}}{\vect{\bar{x}}} &= \vect{\bar{A}}\vect{\bar{x}} + \vect{\bar{D}}\vect{\pert}\text{,}\\
        \vect{\bar{y}} &= \vect{\bar{C}}\vect{\bar{x}}\comma
    \end{split}
\end{align}
with the dynamic matrix
\begin{align}\label{eq:transformed_dynamic_matrix}
    \begin{split}
        \vect{\bar{A}} =
        \resizebox{.9\width}{!}{$
            \left[\begin{array}{ccccc|ccccc|c|ccccc}
                \alpha_{1,1} & 1              & 0       & \cdots  & 0                   & \alpha_{2,1} & 0             & \cdots  & \cdots  & 0                   &        & \alpha_{p,1} & 0              & \cdots      & \cdots  & 0        \\
                \alpha_{1,2} & 0              & 1       & \ddots  & \vdots              & \alpha_{2,2} & \vdots        &         &         & \vdots              &        & \alpha_{p,2} & \vdots         &             &         & \vdots   \\
                \vdots       & \vdots         & \ddots  & \ddots  & 0                   & \vdots       & \vdots        &         &         & \vdots              &        & \vdots       & \vdots         &             &         & \vdots   \\
                \vdots       & \vdots         &         & \ddots  & 1                   & \vdots       & \vdots        &         &         & \vdots              &        & \vdots       & \vdots         &             &         & \vdots   \\
                \vdots       & 0              & \cdots  & \cdots  & 0                   & \vdots       & 0             & \cdots   & \cdots & 0                   &        & \vdots       & 0              & \cdots      & \cdots  & 0        \\
                \hline
                \vdots       & 0              & \cdots  & \cdots  & 0                   & \vdots       & 1             & 0       & \cdots  & 0                   &        & \vdots       & 0              & \cdots      & \cdots  & 0        \\
                \vdots       & \vdots         &         &         & \vdots              & \vdots       & 0             & 1       & \ddots  & \vdots              &        & \vdots       & \vdots         &             &         & \vdots   \\
                \vdots       & \vdots         &         &         & \vdots              & \vdots       & \vdots        & \ddots  & \ddots  & 0                   &        & \vdots       & \vdots         &             &         & \vdots   \\
                \vdots       & 0              & \cdots  & \cdots  & 0                   & \vdots       & \vdots        &         & \ddots  & 1                   &        & \vdots       & \vdots         &             &         & \vdots   \\
                \vdots       & \beta_{1,2,1}  & \cdots  & \cdots  & \beta_{1,2,\mu_1-1} & \vdots       & 0             & \cdots  & \cdots  & 0                   &        & \vdots       & 0              & \cdots      & \cdots  & 0        \\
                \hline
                \vdots       &                &         &         &                     & \vdots       &               &         &         &                     & \ddots & \vdots       &                &             &         &          \\   
                \hline
                \vdots       & 0              & \cdots  & \cdots  & 0                   & \vdots       & 0             & \cdots  & \cdots  & 0                   &        & \vdots       & 1              & 0           & \cdots  & 0        \\
                \vdots       & \vdots         &         &         & \vdots              & \vdots       & \vdots        &         &         & \vdots              &        & \vdots       & 0              & 1           & \ddots  & \vdots   \\
                \vdots       & \vdots         &         &         & \vdots              & \vdots       & \vdots        &         &         & \vdots              &        & \vdots       & \vdots         & \ddots      & \ddots  & 0        \\
                \vdots       & 0              & \cdots  & \cdots  & 0                   & \vdots       & 0             & \cdots  & \cdots  & 0                   &        & \vdots       & \vdots         &             & \ddots  & 1        \\
                \alpha_{1,n} & \beta_{1,p,1}  & \cdots  & \cdots  & \beta_{1,p,\mu_1-1} & \alpha_{2,n} & \beta_{2,p,1} & \cdots  & \cdots  & \beta_{2,p,\mu_2-1} &        & \alpha_{p,n} & 0              & \cdots      & \cdots  & 0        \\
            \end{array}\right]$}
        \comma
    \end{split}
\end{align}
the unknown-input matrix
\begin{align}\label{eq:transformed_disturbance_matrix}
    \begin{split}
        \vect{\bar{D}} =
        \resizebox{.9\width}{!}{$
            \left [ \begin{array}{ccc}
                0 & \cdots & 0 \\
                \vdots & & \vdots \\
                0 & \cdots & 0 \\
                \bar{d}_{\mu_1,1} & \cdots & \bar{d}_{\mu_1,m} \\
                \hline
                0 & \cdots & 0 \\
                \vdots & & \vdots \\
                0 & \cdots & 0 \\
                \bar{d}_{\mu_1+\mu_2,1} & \cdots & \bar{d}_{\mu_1+\mu_2,m} \\
                \hline
                \vdots & & \vdots \\
                \hline
                0 & \cdots & 0 \\
                \vdots & & \vdots \\
                0 & \cdots & 0 \\
                \bar{d}_{n,1} & \cdots & \bar{d}_{n,m}
            \end{array} \right ]
            $}
    \end{split}
\end{align}
and the output matrix
\begin{align}\label{eq:transformed_output_matrix}
    \begin{split}
        \vect{\bar{C}} =
        \resizebox{.9\width}{!}{$
            \left [ \begin{array}{ccccc|ccccc|c|ccccl}
                1      & 0      & \cdots & \cdots  & 0 & 0      & \cdots & \cdots & \cdots & 0      & \cdots & 0      & \cdots & \cdots & \cdots  & 0      \\
                0      & \cdots & \cdots & \cdots  & 0 & 1      & 0      & \cdots & \cdots & 0      & \cdots & 0      & \cdots & \cdots & \cdots  & 0      \\
                \vdots &        &        &         &   &        &        &        &        &        & \ddots &        &        &        &         & \vdots \\   
                \undermat{\mu_1}{0 & \cdots & \cdots & \cdots & 0}  &  \undermat{\mu_2}{0 & \cdots & \cdots & \cdots & 0}  & \cdots  &  \undermat{\mu_p}{1 & 0 & \cdots & \cdots & 0 \textcolor{white}{m} }    
            \end{array} \hspace*{-0.4cm}\right ]
            $}
        .
    \end{split}
\end{align}
\\[18pt]
Furthermore, the orders of the subsystems given by the integers $\mu_j$, $j = 1, \dots, p$, are sorted in descending order, and their sum equals the system order $n$, i.e.,
\begin{align}\label{eq:mu_ordered}
    \mu_1 \geq \mu_2 \geq \dots \geq \mu_p > 0\comma \qquad \sum\limits_{j=1}^{p} \mu_j = n.
\end{align}
\end{theorem}
The proof of Theorem \ref{theorem:existence} and an algorithm which allows for determining suitable transformations are given in Section \ref{sec:transformation_into_canonical_form}.

\begin{remark}
	Note that the proposed observer normal form consists of $p$ coupled single-output systems of orders $\mu_1,\dots,\mu_p$. In the case $p = 1$ the proposed observer normal form \eqref{eq:transformed_system} coincides with the well-known transposed observable canonical form for linear single-output systems.
\end{remark}

\begin{remark}
	The proposed observer normal form exhibits a block lower triangular structure up to an output injection. In contrast Luenberger's observability normal form based on so-called observability indices, the proposed observer normal form also takes the impact of the unknown input $\vect{\pert}$ into account.
\end{remark}

\subsection{Robust observer design in the proposed observer normal form}\label{ssec:robust_observer_design}

In order to robustly reconstruct the state vector $\vect{\bar{x}}$ of the system \eqref{eq:transformed_system}, an observer relying on the RED \citep{Levant1998} is designed. The proposed observer is given by
\begin{subequations}\label{eq:overall_observer}
	\begin{align}\label{eq:robust_observer_of_transformed_system}
	\begin{split}
    \overset{\vect{.}}{\vect{\hat{\bar{x}}}} &= \vect{\bar{A}}\vect{\hat{\bar{x}}} + \vect{\bar{\Pi}} \vect{\sigma}_{\vect{\bar{y}}} + \vect{\bar{l}}(\vect{\sigma}_{\vect{\bar{y}}})  \text{,}\\
	\vect{\hat{\bar{y}}} &= \vect{\bar{C}}\vect{\hat{\bar{x}}} =
	\begin{bmatrix}
	\hat{\bar{x}}_1 & \hat{\bar{x}}_{\mu_1 + 1} & \dots & \hat{\bar{x}}_{\mu_1 + \dots + \mu_{p-1} +1}
	\end{bmatrix}\T
	\comma
	\end{split}
	\end{align}
	where
	\begin{align}\label{eq:output_error}
	\vect{\sigma}_{\vect{\bar{y}}}  = \vect{\bar{y}} - \vect{\hat{\bar{y}}} =
	\begin{bmatrix}
	\sigma_{1} & \sigma_{\mu_1 + 1} & \dots & \sigma_{\mu_1 + \dots + \mu_{p-1} + 1}
	\end{bmatrix}\T
	\end{align}
	is the output error,
	\begin{align}\label{eq:alpha_j_1}
	\vect{\bar{\Pi}} =
	\begin{bmatrix}
	\alpha_{1,1} & \dots & \alpha_{p,1} \\
	\vdots & & \vdots \\
	\alpha_{1,n} & \dots & \alpha_{p,n}
	\end{bmatrix}
	\end{align}
\end{subequations}
provides for a linear output injection in order to compensate for the couplings between the single-output systems and $\vect{\bar{l}}(\vect{\sigma}_{\vect{\bar{y}}})$ is the nonlinear output injection. %
\begin{theorem}\label{THEOREM:OBSERVER}
	Consider the observer \eqref{eq:overall_observer} for the estimation of the state vector $\vect{\bar{x}}$ of system \eqref{eq:transformed_system}, the choice
	\begin{align}\label{eq:nonlinear_output_injection}
	\begin{split}
	\resizebox{.84\width}{!}{$
		\vect{\bar{l}}(\vect{\sigma}_{\vect{\bar{y}}}) =
		\left [ \begin{array}{cccc|c|cccc}
		\kappa_{1,\mu_1-1}\krax{\sigma_1}{\frac{\mu_{1}-1}{\mu_1}} & \cdots & \kappa_{1,1}\krax{\sigma_1}{\frac{1}{\mu_1}} & \kappa_{1,0}\krax{\sigma_1}{0} & \cdots & \kappa_{p,\mu_p-1}\krax{\sigma_{\mu_1+ \dots + \mu_{p-1} +1}}{\frac{\mu_{p}-1}{\mu_p}} & \cdots & \kappa_{p,0}\krax{\sigma_{\mu_1+ \dots + \mu_{p-1} +1}}{0}
		\end{array} \right ]\T
		$}
	\end{split}
	\end{align}
	for the nonlinear output injection and the unknown input $\vect{\pert}$ with bounds given in \eqref{eq:perturbation_bounds}. Then, there exist parameters $\kappa_{j,k}$, $j = 1, \dots, p$, $k = 0, \dots, \mu_j -1$, such that the estimation error $\vect{\sigma} = \vect{\bar{x}} - \vect{\hat{\bar{x}}}$ converges to zero within finite time despite the unknown inputs for any initial states. Moreover, convergence of the observer for all admissible unknown input signals is achieved only if $\kappa_{j,0} > \sum\limits_{i = 1}^{m} L_i \abs{\bar{d}_{\mu_1 + \dots + \mu_{j},i}}$.
\end{theorem}
\begin{proof}[Proof of Theorem \ref{THEOREM:OBSERVER}]
	The estimation error
	\begin{align}\label{eq:estimation_error}
	\vect{\sigma} = \vect{\bar{x}} - \vect{\hat{\bar{x}}} =
	\begin{bmatrix}
	\sigma_{1} & \sigma_{2} & \dots & \sigma_{n}
	\end{bmatrix}\T
	\end{align}
	yields the estimation error dynamics
	\begin{align}\label{eq:estimation_error_dynamics}
	\begin{split}
	&\dot{\vect{\sigma}} = \vect{\bar{A}}\vect{\sigma} - \vect{\bar{\Pi}} \vect{\sigma}_{\vect{\bar{y}}} - \vect{\bar{l}}(\vect{\sigma}_{\vect{\bar{y}}}) +  \vect{\bar{D}}\vect{\pert} \text{,}\\
	&\vect{\sigma}_{\vect{\bar{y}}} = \vect{\bar{C}}\vect{\sigma}.
	\end{split}
	\end{align}
	Taking into account the structure of the involved matrices allows to rewrite the estimation error dynamics \eqref{eq:estimation_error_dynamics} as $p$ coupled subsystems $\Sigma_j$ of order $\mu_j$, i.e.,
	\begin{align}\label{eq:estimation_error_dynamics_element_wise}
	\begin{split}
	\Sigma_1 &
	\begin{cases}
	\dot{\sigma}_1 = \sigma_2 -\kappa_{1,\mu_1-1} \krax{\sigma_{1}}{\frac{\mu_1 - 1}{\mu_1}} \\
	\vdots \\
	\dot{\sigma}_{\mu_1 - 1} = \sigma_{\mu_1} -\kappa_{1,1} \krax{\sigma_{1}}{\frac{1}{\mu_1}} \\
	\dot{\sigma}_{\mu_1} = -\kappa_{1,0} \krax{\sigma_{1}}{0} +\sum\limits_{i = 1}^{m} \bar{d}_{\mu_1,i} \pert_i \\
	\end{cases}\\
	\Sigma_2 &
	\begin{cases}
	\dot{\sigma}_{\mu_1 +1} = \sigma_{\mu_1 + 2} -\kappa_{2,\mu_2-1} \krax{\sigma_{\mu_1 +1}}{\frac{\mu_2 - 1}{\mu_2}} \\
	\vdots \\
	\dot{\sigma}_{\mu_1 + \mu_2 - 1} = \sigma_{\mu_1 + \mu_2} -\kappa_{2,1} \krax{\sigma_{\mu_1 +1}}{\frac{1}{\mu_2}} \\
	\dot{\sigma}_{\mu_1 + \mu_2} = \sum\limits_{l = 1}^{\mu_1 -1} \beta_{1,2,l} \sigma_{l+1} -\kappa_{2,0} \krax{\sigma_{\mu_1 +1}}{0} +\sum\limits_{i = 1}^{m} \bar{d}_{\mu_1 + \mu_2,i} \pert_i \\
	\end{cases} \\
	& \hspace*{0.5cm}\vdots \\
	\Sigma_p &
	\begin{cases}
	\dot{\sigma}_{\mu_1 + \dots + \mu_{p-1} +1} = \sigma_{\mu_1 + \dots + \mu_{p-1} +2} -\kappa_{p,\mu_p-1} \krax{\sigma_{\mu_1 + \dots + \mu_{p-1} +1}}{\frac{\mu_p - 1}{\mu_p}} \\
	\vdots \\
	\dot{\sigma}_{n - 1} = \sigma_{n} -\kappa_{p,1} \krax{\sigma_{\mu_1 + \dots + \mu_{p-1} +1}}{\frac{1}{\mu_p}} \\
	\dot{\sigma}_{n} = \sum\limits_{i = 1}^{p-1} \sum\limits_{l = 1}^{\mu_i -1} \beta_{i,p,l} \sigma_{\mu_1 + \dots + \mu_{i -1} + l+1} - \kappa_{p,0} \krax{\sigma_{\mu_1 + \dots + \mu_{p-1} +1}}{0} +\sum\limits_{i = 1}^{m} \bar{d}_{n,i} \pert_i\text{,}
	\end{cases}
	\end{split}
	\end{align}
	where each subsystem $\Sigma_j$, in terms of structure coincides with the estimation error dynamics of the RED \citep{Levant1998, Levant2003a} with additional couplings $\sum\limits_{i = 1}^{j-1} \sum\limits_{l = 1}^{\mu_i -1} \beta_{i,p,l} \sigma_{\mu_1 + \dots + \mu_{i -1} + l+1}$ in the last differential equation. It is known that there exist parameters {$\kappa_{j,k}$, $k = 0, \dots, \mu_{j -1}$} such that the state variables $\sigma_{\mu_1 + \dots + \mu_{j -1} + 1},\dots,\sigma_{\mu_1 + \dots + \mu_{j}}$ of the respective subsystem $\Sigma_j$ converge to zero in finite time if the right-hand side of the last differential equation is bounded \citep{CRUZZAVALA2016660}. Therefore, since the unknown inputs $\pert_i$ are bounded, the estimation error variables $\sigma_{1}, \dots, \sigma_{\mu_1}$ of subsystem $\Sigma_1$ converge to zero within finite time for properly chosen parameters $\kappa_{1,k}$. During the convergence of the state variables of $\Sigma_1$, the state variables of the other subsystems remain bounded as the systems do not exhibit a finite escape time. After finite-time convergence of the state variables of $\Sigma_1$, the coupling of $\Sigma_1$ and $\Sigma_2$ through the last differential equation of $\Sigma_2$ vanishes and the estimation error variables $\sigma_{\mu_1 + 1},\dots,\sigma_{\mu_1 + \mu_2}$ of system $\Sigma_2$ converge to zero within finite time for properly chosen parameters $\kappa_{2,k}$. This further decouples $\Sigma_3$ from $\Sigma_2$. By induction, it can be shown, that this step-wise finite-time convergence is achieved for all further subsystems $\Sigma_3,\dots,\Sigma_p$. Thus, it can be concluded that the estimation error $\vect{\sigma}$ vanishes in finite time despite the unknown inputs for properly chosen parameters $\kappa_{j,k}$, $j = 1, \dots, p$, $k = 0, \dots, \mu_j -1$.
	
	Furthermore, the robustness of \eqref{eq:estimation_error_dynamics_element_wise} requires that the discontinuity in the last differential equation of each subsystem $\Sigma_j$ is capable of dominating the respective right-hand side \citep{Levant1998}. Since the linear coupling terms in $\Sigma_j$ vanish when all the previous subsystems $\Sigma_1,\dots,\Sigma_{j-1}$ have converged, the condition
	\begin{align}\label{eq:condition_on_kappa_0}
	\kappa_{j,0} > \sum\limits_{i = 1}^{m} L_i \abs{\bar{d}_{\mu_1 + \dots + \mu_{j},i}} \qquad \text{for } j = 1, \dots, p
	\end{align}
	is required.
\end{proof}

\begin{remark}
	The observer parameters $\kappa_{j,k}$ of the subsystem $\Sigma_j$ correspond to the gains of a $(\mu_j -1)^\text{th}$ order RED. It is pointed out that the choice of the observer parameters depends on the bounds of the unknown inputs only and, thus, is independent of the couplings. Since the subsystems are decoupled of each other after some finite time, the parameters of each subsystem can be selected independently. Well-established parameter settings can be found e.g. in \citep[Section 6.7]{Shtessel2013} and \citep{reichhartinger2018arbitrary}.
\end{remark}

\begin{remark}
	In the case $\bar{d}_{\mu_{1} + \dots + \mu_{j},1} = \bar{d}_{\mu_{1} + \dots + \mu_{j},2} =\dots = \bar{d}_{\mu_{1} + \dots + \mu_{j},m} = 0$ for some $j$, i.e., there acts no unknown input on $\Sigma_j$, the output injection $\vect{\bar{l}}(\vect{\sigma}_{\vect{\bar{y}}})$ can be modified such that the corresponding subsystem $\Sigma_j$ coincides with the estimation error dynamics of any other non-robust finite-time differentiator, e.g. the continuous, homogeneous differentiators considered in \citep{Perruquetti4459812, CRUZZAVALA2016660}.
\end{remark}

\begin{remark}
	In the case of a system with a scalar output and a single scalar input, i.e. $p = m = 1$, the proposed observer \eqref{eq:overall_observer} with the output injection \eqref{eq:nonlinear_output_injection} coincides with the observer proposed in \citep{Niederwieser2019} for homogeneity degree $q = -1$.
\end{remark}

\begin{remark}
	If the output is corrupted by additive, uniformly bounded measurement noise ${\vect{v} = \begin{bmatrix} v_1 & \cdots & v_p\end{bmatrix}\T}$, i.e., ${\vect{\bar{y}} = \vect{\bar{C}}\vect{\bar{x}} + \vect{v}}$, $\abs{v_j} \leq \eta_j$, $\eta_j \geq 0$, $j = 1, \dots, p$, the estimation error $\vect{\sigma}$ stays bounded with bounds discussed e.g. in \citep{CRUZZAVALA2016660}.
\end{remark}

\section{Transformation into observer normal form}\label{sec:transformation_into_canonical_form}

The transformation of system \eqref{eq:original_system} into the proposed observer normal form \eqref{eq:transformed_system} is achieved by a regular state transformation of the form
\begin{align}\label{eq:state_transformation}
\vect{\bar{x}} = \vect{T}^{-1} \vect{x} \quad \text{with} \quad \vect{T} \in \mathbb{R}^{n \times n}\comma
\end{align}
and a regular output transformation
\begin{align}\label{eq:output_transformation}
\vect{\bar{y}} = \vect{\Gamma} \vect{y}\quad \text{with} \quad \vect{\Gamma} \in \mathbb{R}^{p \times p}\comma
\end{align}
which yield the matrices
\begin{subequations}\label{eq:matrices_of_transformed_system}
	\begin{align}
	\vect{\bar{A}} &= \vect{T}^{-1}\vect{A}\vect{T}\comma \label{eq:matrices_of_transformed_system_A}\\
	\vect{\bar{D}} &= \vect{T}^{-1}\vect{D}\comma \label{eq:matrices_of_transformed_system_D}\\
	\vect{\bar{C}} &= \vect{\Gamma}\vect{C}\vect{T}\comma \label{eq:matrices_of_transformed_system_C}
	\end{align}
\end{subequations}
of the transformed system.
An algorithm for the construction of the transformation matrices $\vect{T}$ and $\vect{\Gamma}$ is presented in Section \ref{ssec:description_of_proposed_algortihm}. Moreover, a Matlab implementation of the presented algorithm can be downloaded from \url{http://www.reichhartinger.at/}. The theoretical basis including the relevant proofs of Theorem \ref{theorem:existence} is provided in Section \ref{ssec:lemmas_and_theorems_of_transformation}.

\subsection{Description of the proposed algorithm}\label{ssec:description_of_proposed_algortihm}

In order to construct the state transformation matrix $\vect{T}$ and the output transformation matrix $\vect{\Gamma}$, apply the following four-step algorithm:

\subsubsection*{Step 1: Output transformation and output-feedback based decomposition of the dynamic matrix}

\begin{adjustwidth}{1.5em}{0pt}
The first step aims for maximizing the relative degrees of the outputs w.r.t. $\vect{\pert}$ by means of an output transformation and some output-feedback $\vect{\Xi}$. To be more specific, the goal is to find an output transformation \eqref{eq:output_transformation} and a decomposition of the dynamic matrix into
\begin{align}\label{eq:decompose_dynamic_matrix}
\boldsymbol{A} = \vect{\check{A}} - \vect{\Xi}\vect{C} \quad \text{with} \quad \vect{\Xi}\in\mathbb{R}^{n \times p}\comma
\end{align}
such that the auxiliary system
\begin{align}\label{eq:auxiliary_system_step_1}
\begin{split}
&\overset{\vect{.}}{\vect{\check{x}}} = \vect{\check{A}} \vect{\check{x}} + \vect{D}\vect{\Delta}\comma\\
&\vect{\check{y}} = \begin{bmatrix} \check{y}_1 & \check{y}_2 & \dots & \check{y}_p \end{bmatrix}\T = \vect{\check{C}} \vect{\check{x}}\comma
\end{split}
\end{align}
with
\begin{align}\label{eq:auxiliary_system_output_step_1}
\vect{\check{C}} =
\begin{bmatrix}
\vect{\check{c}}_1 &
\vect{\check{c}}_2 &
\dots &
\vect{\check{c}}_p
\end{bmatrix}\T
= \vect{\Gamma}\vect{C}
\end{align}
has the following properties:\\[-10pt]
\begin{enumerate}[label=(\roman*)]
	\item There exist integers ${\mu_1 \geq \mu_{2} \geq \dots \geq \mu_{p} > 0}$ satisfying $\sum\limits_{j=1}^p \mu_j = n$
	such that the $n \times n$ matrix
	\begin{align}\label{eq:reduced_observability_matrix}
	\obsvred =
	\begin{bmatrix}
	\vect{\check{c}}_1\T \\
	\vdots \\
	\vect{\check{c}}_1\T \vect{\check{A}}^{\mu_1 - 1} \\
	\vect{\check{c}}_2\T \\
	\vdots \\
	\vect{\check{c}}_2\T \vect{\check{A}}^{\mu_2 - 1} \\
	\vdots \\
	\vect{\check{c}}_p\T \\
	\vdots \\
	\vect{\check{c}}_p\T \vect{\check{A}}^{\mu_p - 1} \\
	\end{bmatrix}
	\end{align}
	is invertible, i.e.,
	\begin{align}\label{eq:rank_of_reduced_observability_matrix}
	\rank \obsvred = n.
	\end{align}
	\item If the relative degree $\check{\delta}_j$ of the output $\check{y}_j$ w.r.t. the unknown input $\vect{\pert}$ exists, it satisfies ${\check{\delta}_j \geq \mu_j}$, i.e.,
	\begin{align}\label{eq:relative_degree_condition}
	\vect{\check{c}}_j\T \vect{\check{A}}^{i} \vect{D} = \vect{0}\T \qquad \text{for } j = 1,\dots,p\text{, } i = 0,\dots,\mu_{j}-2.
	\end{align}
\end{enumerate}
The following iterative procedure is partially taken from the steps SCB.1, SCB.2 and SCB.3 of the decomposition algorithm proposed in \citep[Section 5.3, pages 119-127]{Chen2004}.

First of all, initialize
\begin{align}\label{eq:init_step_1}
\begin{split}
&\vect{Z} = \vect{C} \comma  \quad
\vect{Z}_j = \vect{c}_j\T \comma  \quad
\vect{\Psi}_j = \vect{e}_{j}\T \in \mathbb{R}^{1 \times p} \comma  \quad
\vect{W} =
\begin{bmatrix}
\vect{w}_1\T \\
\vect{w}_2\T \\
\vdots \\
\vect{w}_p\T
\end{bmatrix}
= 
\begin{bmatrix}
\vect{0}_{1 \times m} \\
\vect{0}_{1 \times m} \\
\vdots \\
\vect{0}_{1 \times m}
\end{bmatrix} \comma \quad
\nu_j = 1 \comma
\end{split}
\end{align}
for all $j = 1,\dots,p$ and the flag vector
\begin{align}\label{eq:init_flag_vector}
\vect{f} =
\begin{bmatrix}
f_1 \\ f_2 \\ \vdots \\ f_p
\end{bmatrix} = 
\begin{bmatrix}
1 \\ 1 \\ \vdots \\ 1
\end{bmatrix}.
\end{align}
Note that these quantities are modified from iteration to iteration and $\vect{Z}$, $\vect{Z}_j$ and $\vect{\Psi}_j$ will be augmented by further rows. Repeat until all elements of $\vect{f}$ have become zero:

For each non-zero element $f_j$ consider the last row $\vect{z}_{j,\nu_j}\T$ of the corresponding matrix
\begin{align}\label{eq:Z_j}
\vect{Z}_j = \begin{bmatrix} \vect{z}_{j,1}\T \\ \vdots \\ \vect{z}_{j,\nu_j}\T \end{bmatrix}.
\end{align}\\
\textbf{Case 1:} If
\begin{align}\label{eq:rank_W_increases}
\rank \begin{bmatrix}
\vect{W}\\
\vect{z}_{j,\nu_j}\T \vect{D}
\end{bmatrix}
> \rank \vect{W}
\end{align}
set $f_j \leftarrow 0$ and $\vect{w}_j\T \leftarrow \vect{z}_{j,\nu_j}\T \vect{D}$. Then, continue with the next non-zero $f_j$.\\
\textbf{Case 2:} If
\begin{align}\label{eq:rank_W_stays}
\rank \begin{bmatrix}
\vect{W}\\
\vect{z}_{j,\nu_j}\T \vect{D}
\end{bmatrix}
= \rank \vect{W}
\end{align}
calculate coefficients $\zeta_{j,\nu_j,k}$, $k = 1,\dots,p$ which allow to represent $\vect{z}_{j,\nu_j}\T \vect{D}$ as a linear combination of the rows $\vect{w}_j\T$ of $\vect{W}$, i.e.,
\begin{align}\label{eq:linear_combinations_of_W}
\vect{z}_{j,\nu_j}\T \vect{D} =
\begin{bmatrix}
\zeta_{j,\nu_j,1} & \dots & \zeta_{j,\nu_j,p}
\end{bmatrix}
\vect{W}.
\end{align}
Note that \eqref{eq:linear_combinations_of_W} may have infinitely many solutions. A reasonable choice is given by the solution with the minimum quadratic norm.
Update the rows of $\vect{Z}_j$ according to
\begin{align}\label{eq:Z_j_update}
\begin{split}
&\vect{z}_{j,1}\T \leftarrow \vect{z}_{j,1}\T - \sum\limits_{k = 1}^{p} \zeta_{j,\nu_j,k} \vect{z}_{k,\nu_k - \nu_j + 1}\T \\
&\vect{z}_{j,2}\T \leftarrow \vect{z}_{j,2}\T - \sum\limits_{k = 1}^{p} \zeta_{j,\nu_j,k} \vect{z}_{k,\nu_k - \nu_j + 2}\T \\
&\vdots \\
&\vect{z}_{j,\nu_j}\T \leftarrow \vect{z}_{j,\nu_j}\T - \sum\limits_{k = 1}^{p} \zeta_{j,\nu_j,k} \vect{z}_{k,\nu_k}\T\comma
\end{split}
\end{align}
where $\vect{z}_{k,l} = \vect{0}\T$ if $l < 1$. Thus, $\vect{Z}_{j}\T \vect{D} = \vect{0}$ is satisfied.
Furthermore, update
\begin{align}\label{eq:Psi_j_update}
&\vect{\Psi}_{j} \leftarrow \vect{\Psi}_{j} - \sum\limits_{k = 1}^{p} \zeta_{j,\nu_j,k}
\begin{bmatrix}
\vect{0}_{(\nu_j - \nu_k) \times p}\\
\vect{\Psi}_{k}
\end{bmatrix}\comma \quad
\vect{Z} \leftarrow
\begin{bmatrix}
    \vect{Z}_1 \\
    \vdots \\
    \vect{Z}_p
\end{bmatrix}.
\end{align}\\
\begin{adjustwidth}{1.5em}{0pt}
\textbf{Sub-case 2.1:} If
\begin{align}\label{eq:Z_no_rank_increase}
& \rank
\begin{bmatrix}
\vect{Z} \\
\vect{z}_{j,\nu_j}\T\vect{A}
\end{bmatrix}
= \rank \vect{Z}
\end{align}
set the corresponding flag $f_j \leftarrow 0$.\\
\textbf{Sub-case 2.2:} If
\begin{align}\label{eq:Z_rank_increase}
& \rank
\begin{bmatrix}
\vect{Z} \\
\vect{z}_{j,\nu_j}\T\vect{A}
\end{bmatrix}
> \rank \vect{Z}
\end{align}
augment the matrices
\begin{align}\label{eq:augment_matrices}
\vect{Z}_j \leftarrow
\begin{bmatrix}
\vect{Z}_j \\
\vect{z}_{j,\nu_j}\T\vect{A}
\end{bmatrix}\comma \quad
\vect{\Psi}_j \leftarrow
\begin{bmatrix}
\vect{\Psi}_j \\
\vect{0}_{1 \times p}
\end{bmatrix}\comma
\end{align}
update
\begin{align}\label{eq:update_Z}
& \vect{Z} \leftarrow
\begin{bmatrix}
	\vect{Z}_1 \\
	\vdots \\
	\vect{Z}_p
\end{bmatrix}\comma
\end{align}
and increase
\begin{align}\label{eq:increase_nu}
& \nu_j \leftarrow \nu_j + 1.
\end{align}
\end{adjustwidth}
This procedure is repeated with the next non-zero element $f_j$ until $\vect{f} = \vect{0}$. Note that, once all elements in the flag vector $\vect{f}$ have become zero, $\vect{Z}$ is an invertible $n \times n$ matrix and $\sum\limits_{j=1}^{p} \nu_{j} = n$ is ensured, which is shown later on. Assign the integers $1,2,\dots,p$ to $j_1, j_2, \dots, j_p$ such that
\begin{align}\label{eq:sorted_index}
\nu_{j_1} \geq \nu_{j_2} \geq \dots \geq \nu_{j_p}
\end{align}
are sorted in descending order and assign the orders $\mu_j$ of the subsystems as
\begin{align}\label{eq:mu_sort_algorithm}
\mu_1 = \nu_{j_1}\comma \quad \mu_2 = \nu_{j_2}\comma\quad \dots\comma \quad \mu_p = \nu_{j_p}.
\end{align}
Consider the matrices $\vect{\Psi}_{j}$ partitioned into their rows
\begin{align}\label{eq:Psi_j}
\vect{\Psi}_j = \begin{bmatrix} \vect{\psi}_{j,1}\T \\ \vdots \\ \vect{\psi}_{j,\nu_j}\T \end{bmatrix}.
\end{align}
and construct the output transformation matrix 
\begin{align}\label{eq:output_transformation_matrix_construction}
& \vect{\Gamma} =
\begin{bmatrix}
\vect{\psi}_{j_1,1}\T \\
\vect{\psi}_{j_2,1}\T \\
\vdots \\
\vect{\psi}_{j_p,1}\T \\
\end{bmatrix}.
\end{align}
Finally, construct the matrix $\vect{\Xi}$ introduced in \eqref{eq:decompose_dynamic_matrix} according to
\begin{align}\label{eq:Xi_construction}
& \vect{\Xi} = \vect{Z}^{-1}
\begin{bmatrix}
\vect{\psi}_{1,2}\T \\
\vdots \\
\vect{\psi}_{1,\nu_1}\T \\
\vect{0}\T \\ \hline
\vdots \\ \hline
\vect{\psi}_{p,2}\T \\
\vdots \\
\vect{\psi}_{p,\nu_p}\T \\
\vect{0}\T \\
\end{bmatrix}\comma
\end{align}
and calculate the matrices $\vect{\check{A}}$ and $\vect{\check{C}}$ of the auxiliary system \eqref{eq:auxiliary_system_step_1} from \eqref{eq:decompose_dynamic_matrix} and \eqref{eq:auxiliary_system_output_step_1}, respectively.
\end{adjustwidth}

\subsubsection*{Step 2: Calculation of $p$ columns of the state transformation matrix}

\begin{adjustwidth}{1.5em}{0pt}
	
	The construction of the state transformation \eqref{eq:state_transformation} is based on the auxiliary system \eqref{eq:auxiliary_system_step_1}. Consider the transformation matrix $\vect{T}$ to be expressed by its column vectors $\vect{t}_i$ i.e.,
	\begin{align}\label{eq:transformation_matrix}
	\vect{T} =
	\left [ \begin{array}{cccc|ccc|c|ccc}
	\vect{t}_1 & \vect{t}_2 & \dots & \vect{t}_{\mu_1} & t_{\mu_1+1} & \dots & \vect{t}_{\mu_1+ \mu_2} & \dots & \vect{t}_{\mu_1 + \dots + \mu_{p-1}+1} & \dots & \vect{t}_{n}
	\end{array} \right ].
	\end{align}
	Calculate the column vectors $\vect{t}_{\mu_1}$, $\vect{t}_{\mu_1 + \mu_2}$, $\dots$, $\vect{t}_{n}$ of the transformation matrix $\vect{T}$ from
	\begin{align}\label{eq:t_mu_j_major}
	\begin{bmatrix}
	\vect{t}_{\mu_1} & \vect{t}_{\mu_1 + \mu_2} & \dots & \vect{t}_{n}
	\end{bmatrix}
	=
	\obsvred^{-1}
	\begin{bmatrix}
	\vect{e}_{\mu_1} & \vect{e}_{\mu_1 + \mu_2} & \dots & \vect{e}_{n}
	\end{bmatrix}\text{,}
	\end{align}
	where $\obsvred$ is the observability matrix like matrix of the auxiliary system given in \eqref{eq:reduced_observability_matrix}.
\end{adjustwidth}

\subsubsection*{Step 3: Calculation of the coefficients $\beta_{j,k,l}$}

\begin{adjustwidth}{1.5em}{0pt}
	For each ${j = 1,2,\dots,p-1}$ calculate the coefficients $b_{j,k,l}$ of the matrix $\vect{\bar{A}}$ for ${k = j+1,j+2, \dots, p}$ and ${l = 1,2 \dots, \mu_j -1}$ in the following way:
	\begin{enumerate}
		\item Construct the invertible matrix
		\begin{align}\label{eq:H_beta_j}
		\Abeta^{(j)} =
		\begin{bmatrix}
		\Abeta^{(j)}_{j+1,j+1} & \cdots & \Abeta^{(j)}_{j+1,p} \\
		\vdots& \ddots & \vdots \\
		\Abeta^{(j)}_{p,j+1} & \cdots & \Abeta^{(j)}_{p,p} \\
		\end{bmatrix}\text{,}
		\end{align}
		where the submatrices $\Abeta^{(j)}_{r,s} \in \mathbb{R}^{(\mu_j-\mu_r) \times (\mu_j-\mu_s)}$ have the Toeplitz structure
		\begin{align}\label{eq:H_beta_j_k_l}
		\resizebox{.84\width}{!}{$
			\Abeta^{(j)}_{r,s} =
			\begin{cases}
			\hspace{0.1cm}
			\begin{bmatrix}
			\vect{\check{c}}_r\T \vect{\check{A}}^{\mu_s-1}\vect{t}_{\mu_1+ \dots +\mu_s}  & \vect{\check{c}}_r\T \vect{\check{A}}^{\mu_s}\vect{t}_{\mu_1+ \dots +\mu_s}   & \cdots & \vect{\check{c}}_r\T \vect{\check{A}}^{\mu_j-2}\vect{t}_{\mu_1+ \dots +\mu_s} \\
			\vect{\check{c}}_r\T \vect{\check{A}}^{\mu_s-2}\vect{t}_{\mu_1+ \dots +\mu_s}  & \vect{\check{c}}_r\T \vect{\check{A}}^{\mu_s-1}\vect{t}_{\mu_1+ \dots +\mu_s} & \ddots & \vect{\check{c}}_r\T \vect{\check{A}}^{\mu_j-3}\vect{t}_{\mu_1+ \dots +\mu_s} \\
			\vdots                                                         & \ddots                                                        & \ddots & \vdots \\
			\vect{\check{c}}_r\T \vect{\check{A}}^{\mu_r} \vect{t}_{\mu_1+ \dots +\mu_s}   & \ddots                                                        & \ddots & \vdots \\
			0                                                              & \ddots                                                        & \ddots & \vdots \\
			\vdots                                                         & \ddots                                                        & \ddots & \vdots \\
			\vdots                                                         &                                                               & \ddots & \vect{\check{c}}_r\T \vect{\check{A}}^{\mu_r} \vect{t}_{\mu_1+ \dots +\mu_s} \\
			0                                                              & \cdots                                                        & \cdots & 0
			\end{bmatrix}
			\qquad \text{if} \quad r > s \quad (\Rightarrow \mu_r \leq \mu_s)\text{,}
			\vspace{0.2cm} \\ \hspace{0.1cm}
			\begin{bmatrix}
			1                                                              & \vect{\check{c}}_r\T \vect{\check{A}}^{\mu_r}\vect{t}_{\mu_1+ \dots +\mu_r}   & \cdots & \vect{\check{c}}_r\T \vect{\check{A}}^{\mu_j-2}\vect{t}_{\mu_1+ \dots +\mu_r} \\
			0                                                              & \ddots                                                        & \ddots & \vdots \\
			\vdots                                                         & \ddots                                                        & \ddots & \vect{\check{c}}_r\T \vect{\check{A}}^{\mu_r}\vect{t}_{\mu_1+ \dots +\mu_r} \\
			0                                                              & \cdots                                                        & 0      & 1
			\end{bmatrix}
			\qquad \text{if} \quad r = s \quad (\Rightarrow \mu_r = \mu_s) \text{,}
			\vspace{0.2cm} \\ \hspace{0.1cm}
			\begin{bmatrix}
			0                                                              & \cdots & 0      & \vect{\check{c}}_r\T \vect{\check{A}}^{\mu_r}\vect{t}_{\mu_1+ \dots +\mu_s}   & \cdots & \vect{\check{c}}_r\T \vect{\check{A}}^{\mu_j-2}\vect{t}_{\mu_1+ \dots +\mu_s} \\
			\vdots                                                         &        & \vdots & \ddots                                                        & \ddots & \vdots \\
			\vdots                                                         &        & \vdots &                                                               & \ddots & \vect{\check{c}}_r\T \vect{\check{A}}^{\mu_r}\vect{t}_{\mu_1+ \dots +\mu_s} \\
			0                                                              & \cdots & 0      & \cdots                                                        & \cdots & 0
			\end{bmatrix}\qquad \text{if} \quad r < s \quad (\Rightarrow \mu_r \geq \mu_s).
			\end{cases}$}
		\end{align}
		\item Build up the vector
		\begin{align}\label{eq:bbeta_j}
		\vect{\bbeta}^{(j)} =
		\begin{bmatrix}
		\vect{\bbeta}^{(j)}_{j+1} \\
		\vect{\bbeta}^{(j)}_{j+2} \\
		\vdots \\
		\vect{\bbeta}^{(j)}_{p}
		\end{bmatrix}\comma \quad \text{where} \quad
    	\vect{\bbeta}^{(j)}_r =
        \begin{bmatrix}
        \vect{\check{c}}_r\T \vect{\check{A}}^{\mu_j-1}\vect{t}_{\mu_1+ \dots +\mu_j}\\
        \vect{\check{c}}_r\T \vect{\check{A}}^{\mu_j-2}\vect{t}_{\mu_1+ \dots +\mu_j}\\
        \vdots \\
        \vect{\check{c}}_r\T \vect{\check{A}}^{\mu_r}\vect{t}_{\mu_1+ \dots +\mu_j}
         \end{bmatrix}.
		\end{align}
		\item Solve the linear system of equations
		\begin{align}\label{eq:beta_equation_sys}
		\Abeta^{(j)}\vect{\beta}^{(j)} = \bbeta^{(j)}
		\end{align}
		for the vector $\vect{\beta}^{(j)}$ which holds the coefficients $\beta_{j,k,l}$ for $l \geq \mu_k$, i.e.,
		\begin{align}\label{eq:beta_j}
		\vect{\beta}^{(j)} =
		\begin{bmatrix}
		\vect{\beta}^{(j)}_{j+1} \\
		\vect{\beta}^{(j)}_{j+2} \\
		\vdots \\
		\vect{\beta}^{(j)}_{p}
		\end{bmatrix}\text{,}
		\qquad
		\text{with}
		\quad
		\vect{\beta}^{(j)}_r =
		\begin{bmatrix}
		\beta_{j,r,\mu_{r}} \\
		\vdots \\
		\beta_{j,r,\mu_{j}-1}
		\end{bmatrix}.
		\end{align}
		Set the remaining coefficients to zero, i.e.,
		\begin{align}\label{eq:beta_j_0}
		\beta_{j,k,l} = 0 \qquad \text{for} \quad l < \mu_k.
		\end{align}		
	\end{enumerate}
\end{adjustwidth}

\subsubsection*{Step 4: Construction of the state transformation matrix}

\begin{adjustwidth}{1.5em}{0pt}
	Calculate the remaining columns of the transformation matrix $\vect{T}$ given in \eqref{eq:transformation_matrix} according to
	\begin{align}\label{eq:t_i}
	\vect{t}_{\mu_1+ \dots + \mu_j - i} = \vect{\check{A}}^{i}\vect{t}_{\mu_1 + \dots + \mu_j} - \sum\limits_{r = j+1}^{p}\sum\limits_{q=1}^{i} \beta_{j,r,\mu_j - q} \vect{\check{A}}^{i-q}\vect{t}_{\mu_1+ \dots +\mu_r}\comma
	\end{align}
	where $j = 1,2,\dots,p$ and $i = 0,1,\dots,\mu_j-1$.
	Note that \eqref{eq:t_i} is also consistent in the case $i = 0$ which is exploited in the proofs in the Appendix. \\[-5pt]
\end{adjustwidth}

\noindent The presented algorithm yields non-singular transformation matrices $\vect{T}$ and $\vect{\Gamma}$ whenever the original system \eqref{eq:original_system} is strongly observable, which is shown in the following Section \ref{ssec:lemmas_and_theorems_of_transformation}.

\subsection{Existence of the proposed state transformation}\label{ssec:lemmas_and_theorems_of_transformation}

In this section, the theoretical basis for the previously presented state transformation is established. It shown that the algorithm proposed in Section \ref{ssec:description_of_proposed_algortihm} yields a description of the system in the observer normal form proposed in Theorem \ref{theorem:existence} which finally proves this Theorem.

\begin{proof}[Proof of Theorem \ref{theorem:existence}]
	In order to prove Theorem \ref{theorem:existence}, the following lemma is useful:
	\begin{lemma}\label{LEM:TRANSFORMATION_LEMMA}
		Under the conditions given in Theorem \ref{theorem:existence}, the following statements are true:
		\begin{enumerate}[label=\alph*)]
			\item The auxiliary system \eqref{eq:auxiliary_system_step_1} generated by Step 1 of the  algorithm proposed in Section \ref{ssec:description_of_proposed_algortihm} satisfies the conditions \eqref{eq:rank_of_reduced_observability_matrix} and \eqref{eq:relative_degree_condition}.
			\item The transformation matrix $\vect{T}$ constructed by the proposed algorithm in Section \ref{ssec:description_of_proposed_algortihm} is guaranteed to be non-singular, regardless of the specific values of $\beta_{j,k,l}$.%
			\item There exists a unique solution of the system of equations \eqref{eq:beta_equation_sys} in Step 3 of the algorithm proposed in Section \ref{ssec:description_of_proposed_algortihm}.
			\item If the transformation matrix $\vect{T}$ is constructed according to the algorithm in Section \ref{ssec:description_of_proposed_algortihm}, then the dynamic matrix $\vect{\bar{A}}$ of the transformed system \eqref{eq:transformed_system} given in \eqref{eq:matrices_of_transformed_system_A} takes the proposed form \eqref{eq:transformed_dynamic_matrix}.
			\item If the transformation matrix $\vect{T}$ is constructed according to the algorithm in Section \ref{ssec:description_of_proposed_algortihm}, then the unknown-input matrix $\vect{\bar{D}}$ of the transformed system \eqref{eq:transformed_system} given in \eqref{eq:matrices_of_transformed_system_D} takes the proposed form \eqref{eq:transformed_disturbance_matrix}.
			\item If the transformation matrix $\vect{T}$ is constructed according to the algorithm in Section \ref{ssec:description_of_proposed_algortihm}, then the output matrix $\vect{\bar{C}}$ of the transformed system \eqref{eq:transformed_system} given in \eqref{eq:matrices_of_transformed_system_C} takes the proposed form \eqref{eq:transformed_output_matrix}.
		\end{enumerate}	
	\end{lemma}
	The proof of Lemma \ref{LEM:TRANSFORMATION_LEMMA} is provided in Appendix \ref{app:transformation_lemma}. Note that Lemma \ref{LEM:TRANSFORMATION_LEMMA}a) is a prerequisite for the proofs of Lemma \ref{LEM:TRANSFORMATION_LEMMA}b) to \ref{LEM:TRANSFORMATION_LEMMA}f). Lemma \ref{LEM:TRANSFORMATION_LEMMA}b) ensures $\vect{T}$ to be non-singular. Moreover, $\vect{\Gamma}$ is ensured to be non-singular which follows directly from the results presented in \citep[Section 5.3, pages 119-127]{Chen2004}.
	Lemma \ref{LEM:TRANSFORMATION_LEMMA}c) states that a unique solution of the system of equations \eqref{eq:beta_equation_sys} exists, which ensures the existence of the coefficients $\beta_{j,k,l}$ of the dynamic matrix $\vect{\bar{A}}$ of the transformed system. Finally, from Lemma \ref{LEM:TRANSFORMATION_LEMMA}d), \ref{LEM:TRANSFORMATION_LEMMA}e) and \ref{LEM:TRANSFORMATION_LEMMA}f) it follows directly that the transformed system takes the proposed observer normal form, which completes the proof of Theorem \ref{theorem:existence}.
\end{proof}

\section{Example}\label{sec:Simulation_examples}

In order to demonstrate the effectiveness of the proposed method, it is applied to a linearized model of the lateral motion of a light aircraft taken from \citep{Mudge1988}. Note that another finite-time approach for the robust estimation of the state variables is already applied to this model in \citep{Floquet2006}. However, this approach requires bounded state variables in order to ensure convergence of the observer.

The system of order $n = 7$ with $2$ control inputs, $m = 1$ unknown input and $p = 2$ outputs is given by
\begin{align}\label{eq:tutorial_original_system}
\begin{split}
\dot{\vect{x}} &=  \vect{A}\vect{x} + \vect{B}\vect{u} + \vect{D}\pert\text{,}\\
\vect{y} &= \vect{C}\vect{x}\comma
\end{split}
\end{align}
where
\begin{align}\label{eq:tutorial_original_system_matrices}
\resizebox{.84\width}{!}{$
\vect{A} =
\begin{bmatrix}
-0.3 & 0 & -33 & 9.81 & 0 & -5.4 & 0 \\
-0.1 & -8.3 & 3.75 & 0 & 0 & 0 & -28.6 \\
0.37 & 0 & -0.64 & 0 & 0 & -9.5 & 0 \\
0 & 1 & 0 & 0 & 0 & 0 & 0 \\
0 & 0 & 1 & 0 & 0 & 0 & 0 \\
0 & 0 & 0 & 0 & 0 & -10 & 0 \\
0 & 0 & 0 & 0 & 0 & 0 & -5 \\
\end{bmatrix}\text{,}\quad
\vect{B} =
\begin{bmatrix}
0 & 0 \\
0 & 0 \\
0 & 0 \\
0 & 0 \\
0 & 0 \\
20 & 0 \\
0 & 10
\end{bmatrix}\text{,}\quad
\vect{D} =
\begin{bmatrix}
0 \\ 0 \\ 0 \\ 0 \\ 0 \\ 20 \\ 0
\end{bmatrix}\text{,}\quad
\vect{C} =
\begin{bmatrix}
0 & 1 & 0 & 0 & 0 & 0 & 0 \\
0 & 0 & 0 & 0 & 1 & 0 & 0
\end{bmatrix}.$}
\end{align}
The state vector ${\vect{x} = \begin{bmatrix} v & p & r & \phi & \psi & \zeta & \xi \end{bmatrix}\T}$ consists of the sideslip velocity $v$, the roll rate $p$, the yaw rate $r$, the roll angle $\phi$, the yaw angle $\psi$, the rudder angle $\zeta$ and the aileron angle $\xi$. The control input $\vect{u} = \begin{bmatrix} \zeta_c & \xi_c \end{bmatrix}\T$ is given by the rudder angle demand $\zeta_c$ and the aileron angle demand $\xi_c$. For the unknown input $\pert$ a bounded actuator fault in the rudder is considered. The output $\vect{y} = \begin{bmatrix} y_1 & y_2 \end{bmatrix}\T$ provides measurements of the roll rate $m$ and the yaw angle $\psi$. Note that the considered system is strongly observable which can be shown by means of the rank condition \eqref{eq:strong_observability_condition_rosenbrock} regarding the Rosenbrock matrix. Furthermore, the considered system \eqref{eq:tutorial_original_system} is unstable since $\vect{A}$ has two eigenvalues with non-negative real part located at $s_1 = 0$ and $s_2 = 0.1219$. For this reason, all the derivatives of the output may be unbounded and, thus, a direct application of the RED as a state observer without any prestabilization of the estimation error dynamics is not possible.

In the following, the algorithm proposed in Section \ref{ssec:description_of_proposed_algortihm} is applied to system \eqref{eq:tutorial_original_system} in order to provide a system description in the presented observer normal form. A robust observer is designed according to Section \ref{ssec:robust_observer_design} in order to provide exact estimates of the state vector in the presence of an unknown input $\pert$. The efficiency of this observer is confirmed by numerical simulations. Furthermore, a straightforward extension of the proposed observer is used in order to reconstruct the unknown input $\pert$ which is incorporated into a robust state-feedback control law in order to asymptotically stabilize the plant in the origin despite the unknown input $\pert$.

\subsection{Transformation into the proposed observer normal form}\label{ssec:tutorial_transformation_into_observability_form}

In the following, the construction of the state transformation matrix $\vect{T}$ and the output transformation matrix $\vect{\Gamma}$ based on the algorithm presented in Section \ref{ssec:description_of_proposed_algortihm} is illustrated step-by-step. It is pointed out that all numbers are presented rounded to four decimal places.

\subsubsection*{Step 1: Output transformation and output-feedback based decomposition of the dynamic matrix}

\begin{adjustwidth}{1.5em}{0pt}
First of all, the initialization of the required quantities yields
\begin{align}\label{eq:init_step_1_example}
\begin{split}
&\vect{Z} = \vect{C} = \begin{bmatrix}
0 & 1 & 0 & 0 & 0 & 0 & 0 \\
0 & 0 & 0 & 0 & 1 & 0 & 0
\end{bmatrix} \comma  \quad
\vect{Z}_1 = \vect{c}_1\T = \begin{bmatrix}
0 & 1 & 0 & 0 & 0 & 0 & 0 
\end{bmatrix} \comma  \\
&\vect{Z}_2 = \vect{c}_2\T = \begin{bmatrix}
0 & 0 & 0 & 0 & 1 & 0 & 0
\end{bmatrix} \comma \quad
\vect{\Psi}_1 = \vect{e}_1\T = \begin{bmatrix}
1 & 0
\end{bmatrix} \comma  \quad
\vect{\Psi}_2 = \vect{e}_2\T = \begin{bmatrix}
0 & 1
\end{bmatrix} \comma  \\
&\vect{W} =
\begin{bmatrix}
w_1 \\
w_2
\end{bmatrix}
= 
\begin{bmatrix}
0 \\
0 \\
\end{bmatrix} \comma \quad
\nu_1 = 1 \comma \quad
\nu_2 = 1 \comma \quad
\vect{f} =
\begin{bmatrix}
f_1 \\ f_2
\end{bmatrix} = 
\begin{bmatrix}
1 \\ 1
\end{bmatrix}.
\end{split}
\end{align}
\textbf{Iteration 1:} Starting with the first non-zero flag ${f_1 = 1}$ leads to Case 2 since ${\vect{z}_{1,1}\T \vect{D} = 0}$. Thus, ${\zeta_{1,1,1} = \zeta_{1,1,2} = 0}$ and $\vect{z}_{1,1}$ and $\vect{\Psi}_1$ remain unchanged. Then, Sub-case 2.2 is entered and the updates
\begin{align}\label{eq:example_update_1_1}
\begin{split}
&\vect{Z}_1 \leftarrow
\begin{bmatrix}
\vect{Z}_1 \\
\vect{z}_{1,1}\T\vect{A}
\end{bmatrix}
=
\begin{bmatrix}
0 & 1 & 0 & 0 & 0 & 0 & 0 \\
-0.1 & -8.3 & 3.75 & 0 & 0 & 0 & -28.6 
\end{bmatrix} \comma  \quad
\vect{\Psi}_1 \leftarrow
\begin{bmatrix}
\vect{\Psi}_1 \\
\vect{0} \T
\end{bmatrix}
=
\begin{bmatrix}
1 & 0 \\
0 & 0
\end{bmatrix} \comma  \\
&\vect{Z} \leftarrow
\begin{bmatrix}
\vect{Z}_1 \\ \vect{Z}_2
\end{bmatrix}
\begin{bmatrix}
0 & 1 & 0 & 0 & 0 & 0 & 0 \\
-0.1 & -8.3 & 3.75 & 0 & 0 & 0 & -28.6 \\
0 & 0 & 0 & 0 & 1 & 0 & 0
\end{bmatrix} \comma  \quad
\nu_1 \leftarrow \nu_1 + 1 = 2\comma
\end{split}
\end{align}
are carried out. 

\noindent\textbf{Iteration 2:} Continue with the non-zero flag ${f_2 = 1}$. Again, ${\vect{z}_{2,1}\T \vect{D} = 0}$ leads to Case 2, where ${\zeta_{2,1,1} = \zeta_{2,1,2} = 0}$ and $\vect{z}_{2,1}$ and $\vect{\Psi}_2$ remain unchanged. Sub-case 2.2 occurs and the updates
\begin{align}\label{eq:example_update_2_1}
\begin{split}
&\vect{Z}_2 \leftarrow
\begin{bmatrix}
\vect{Z}_2 \\
\vect{z}_{2,1}\T \vect{A}
\end{bmatrix}
=
\begin{bmatrix}
0 & 0 & 0 & 0 & 1 & 0 & 0 \\
0 & 0 & 1 & 0 & 0 & 0 & 0 
\end{bmatrix} \comma  \quad
\vect{\Psi}_2 \leftarrow
\begin{bmatrix}
\vect{\Psi}_2 \\ \vect{0}\T
\end{bmatrix}
=
\begin{bmatrix}
0 & 1 \\
0 & 0
\end{bmatrix} \comma  \\
&\vect{Z} \leftarrow
\begin{bmatrix}
\vect{Z}_1 \\ \vect{Z}_2
\end{bmatrix}
=
\begin{bmatrix}
0 & 1 & 0 & 0 & 0 & 0 & 0 \\
-0.1 & -8.3 & 3.75 & 0 & 0 & 0 & -28.6 \\
0 & 0 & 0 & 0 & 1 & 0 & 0 \\
0 & 0 & 1 & 0 & 0 & 0 & 0
\end{bmatrix} \comma  \quad
\nu_2 \leftarrow \nu_2 + 1 = 2\comma
\end{split}
\end{align}
are performed.

\noindent\textbf{Iterations 3 \& 4:} Both, ${\vect{z}_{1,2}\T \vect{D} = 0}$ and ${\vect{z}_{2,2}\T \vect{D} = 0}$ and, again, Case 2 and Sub-case 2.2 are entered which yields the updates
\begin{align}\label{eq:example_update_j_2}
\begin{split}
&\vect{Z}_1 \leftarrow 
\begin{bmatrix}
\vect{Z}_1 \\
\vect{z}_{1,2}\T\vect{A}
\end{bmatrix}
=
\begin{bmatrix}
0 & 1 & 0 & 0 & 0 & 0 & 0 \\
-0.1 & -8.3 & 3.75 & 0 & 0 & 0 & -28.6 \\
2.2475 & 68.89 & -30.225 &  -0.981 & 0 & -35.085 & 380.38
\end{bmatrix} \comma  \\
&\vect{\Psi}_1 \leftarrow
\begin{bmatrix}
\vect{\Psi}_1 \\ \vect{0}\T
\end{bmatrix}
=
\begin{bmatrix}
1 & 0 \\
0 & 0 \\
0 & 0
\end{bmatrix} \comma  \\
&\vect{Z}_2 \leftarrow
\begin{bmatrix}
\vect{Z}_2 \\
\vect{z}_{2,2}\T\vect{A}
\end{bmatrix}
=
\begin{bmatrix}
0 & 0 & 0 & 0 & 1 & 0 & 0 \\
0 & 0 & 1 & 0 & 0 & 0 & 0 \\
0.37 & 0 & -0.64 & 0 & 0 & -9.5 & 0
\end{bmatrix} \comma  \quad
\vect{\Psi}_2 \leftarrow
\begin{bmatrix}
\vect{\Psi}_2 \\ \vect{0}\T
\end{bmatrix}
=
\begin{bmatrix}
0 & 1 \\
0 & 0 \\
0 & 0
\end{bmatrix} \comma  \\
&\vect{Z} \leftarrow
\begin{bmatrix}
\vect{Z}_1 \\ \vect{Z}_2
\end{bmatrix}
= \begin{bmatrix}
0 & 1 & 0 & 0 & 0 & 0 & 0 \\
-0.1 & -8.3 & 3.75 & 0 & 0 & 0 & -28.6 \\
2.2475 & 68.89 & -30.225 &  -0.981 & 0 & -35.085 & 380.38 \\
0 & 0 & 0 & 0 & 1 & 0 & 0 \\
0 & 0 & 1 & 0 & 0 & 0 & 0 \\
0.37 & 0 & -0.64 & 0 & 0 & -9.5 & 0
\end{bmatrix} \comma  \\
&\nu_1 \leftarrow \nu_1 + 1 = 3\comma \quad \nu_2 \leftarrow \nu_2 + 1 = 3.
\end{split}
\end{align}
\noindent\textbf{Iteration 5:} Again, the non-zero flag ${f_1 = 1}$ is considered. Since ${\vect{z}_{1,3}\T \vect{D} = -701.7}$ increases the rank of $\vect{W}$, Case 1 is applied, i.e.,
\begin{align}\label{eq:example_update_1_3}
f_1 \leftarrow 0\comma \quad
\vect{W} \leftarrow
\begin{bmatrix}
	w_1 \\
	w_2
\end{bmatrix}
= 
\begin{bmatrix}
	-701.7 \\
	0 \\
\end{bmatrix} .
\end{align}
\noindent\textbf{Iteration 6:} The flag ${f_2 = 1}$ is the only remaining non-zero flag. Considering ${\vect{z}_{2,3}\T \vect{D} = -190}$ leads to Case 2, since $\vect{z}_{2,3}\T \vect{D}$ can be expressed in terms of the rows of $\vect{W}$, i.e., ${\vect{z}_{2,3}\T \vect{D} = \zeta_{2,3,1} w_1 + \zeta_{2,3,2}w_2}$, where
\begin{align}\label{eq:example_zeta_2_3}
\zeta_{2,3,1} = 0.2708\comma \quad \zeta_{2,3,2} = 0.
\end{align}
Updating the rows of $\vect{Z}_2$ yields
\begin{align}\label{eq:example_Z_2_update}
\begin{split}
&\vect{z}_{2,1}\T \leftarrow \vect{z}_{2,1}\T - \zeta_{2,3,1} \vect{z}_{1,1}\T =
\begin{bmatrix}
0 & -0.2708 & 0 & 0 & 1 & 0 & 0
\end{bmatrix} \\
&\vect{z}_{2,2}\T \leftarrow \vect{z}_{2,2}\T - \zeta_{2,3,1} \vect{z}_{1,2}\T =
\begin{bmatrix}
0.0271 & 2.2474 & -0.0154 & 0 & 0 & 0 & 7.7441
\end{bmatrix} \\
&\vect{z}_{2,3}\T \leftarrow \vect{z}_{2,3}\T - \zeta_{2,3,1} \vect{z}_{1,3}\T =
\begin{bmatrix}
-0.2386 & -18.6534 & 7.5441 & 0.2656 & 0 & 0 & -102.9959
\end{bmatrix}.
\end{split}
\end{align}
Furthermore,
\begin{align}\label{eq:example_Psi_2_update}
&\vect{\Psi}_{2} \leftarrow \vect{\Psi}_{2} - \zeta_{2,3,1}\vect{\Psi}_{1} =
\begin{bmatrix}
-0.2708 & 1 \\
0 & 0 \\
0 & 0
\end{bmatrix}
\end{align}
is updated. Then, Sub-case 2.2 occurs, i.e.,
\begin{align}\label{eq:example_update_2_3}
\begin{split}
&\vect{Z}_2 \leftarrow 
\begin{bmatrix}
\vect{Z}_2 \\
\vect{z}_{2,3}\T\vect{A}
\end{bmatrix}
=
\begin{bmatrix}
0 & -0.2708 & 0 & 0 & 1 & 0 & 0 \\
0.0271 & 2.2474 & -0.0154 & 0 & 0 & 0 & 7.7441 \\
-0.2386 & -18.6534 & 7.5441 & 0.2656 & 0 & 0 & -102.9959 \\
4.7282 & 155.089 & -66.9061 & -2.3403 & 0 & -70.3803 & 1048.467
\end{bmatrix} \comma  \\
&\vect{\Psi}_2 \leftarrow
\begin{bmatrix}
\vect{\Psi}_2 \\
\vect{0}\T
\end{bmatrix}
=
\begin{bmatrix}
-0.2708 & 1 \\
0 & 0 \\
0 & 0 \\
0 & 0
\end{bmatrix}\comma  \\
&\vect{Z} \leftarrow
\begin{bmatrix}
\vect{Z}_1 \\ \vect{Z}_2
\end{bmatrix}
=
\begin{bmatrix}
0 & 1 & 0 & 0 & 0 & 0 & 0 \\
-0.1 & -8.3 & 3.75 & 0 & 0 & 0 & -28.6 \\
2.2475 & 68.89 & -30.225 &  -0.981 & 0 & -35.085 & 380.38 \\
0 & -0.2708 & 0 & 0 & 1 & 0 & 0 \\
0.0271 & 2.2474 & -0.0154 & 0 & 0 & 0 & 7.7441 \\
-0.2386 & -18.6534 & 7.5441 & 0.2656 & 0 & 0 & -102.9959 \\
4.7282 & 155.089 & -66.9061 & -2.3403 & 0 & -70.3803 & 1048.467
\end{bmatrix} \comma  \\
&\nu_2 \leftarrow \nu_2 + 1 = 4\comma
\end{split}
\end{align}

\noindent\textbf{Iteration 7:} The flag ${f_2 = 1}$ is still non-zero. Hence, ${\vect{z}_{2,4}\T\vect{D} = -1407.6058}$ is examined which results in Case 2 since ${\vect{z}_{2,4}\T\vect{D} = \zeta_{2,4,1} w_1 + \zeta_{2,4,2}w_2}$ with
\begin{align}\label{eq:example_zeta_2_4}
\zeta_{2,4,1} = 2.006\comma \quad \zeta_{2,4,2} = 0.
\end{align}
The rows of $\vect{Z}_2$ are updated as
\begin{align}\label{eq:example_Z_2_update_4}
\begin{split}
&\vect{z}_{2,2}\T \leftarrow \vect{z}_{2,2}\T - \zeta_{2,4,1} \vect{z}_{1,1}\T =
\begin{bmatrix}
0.0271 & 0.2414 & -0.0154 & 0 & 0 & 0 & 7.7441
\end{bmatrix} \\
&\vect{z}_{2,3}\T \leftarrow \vect{z}_{2,3}\T - \zeta_{2,4,1} \vect{z}_{1,2}\T =
\begin{bmatrix}
-0.038 & -2.0037 & 0.0216 & 0.2656 & 0 & 0 & -45.6244
\end{bmatrix} \\
&\vect{z}_{2,4}\T \leftarrow \vect{z}_{2,4}\T - \zeta_{2,4,1} \vect{z}_{1,3}\T =
\begin{bmatrix}
0.2197 & 16.896 & -6.2749 & -0.3724 & 0 & 0 & 285.427
\end{bmatrix}
\end{split}
\end{align}
and
\begin{align}\label{eq:example_Psi_2_update_4}
&\vect{\Psi}_{2} \leftarrow \vect{\Psi}_{2} - \zeta_{2,4,1}\begin{bmatrix}
\vect{0}\T \\
\vect{\Psi}_{1}
\end{bmatrix} =
\begin{bmatrix}
-0.2708 & 1 \\
-2.006 & 0 \\
0 & 0 \\
0 & 0
\end{bmatrix}.
\end{align}
Then, Sub-case 2.1 is entered and, hence,
\begin{align}\label{eq:example_Z_final}
&\vect{Z} \leftarrow
\begin{bmatrix}
	\vect{Z}_1 \\ \vect{Z}_2
\end{bmatrix}
=
\begin{bmatrix}
0 & 1 & 0 & 0 & 0 & 0 & 0 \\
-0.1 & -8.3 & 3.75 & 0 & 0 & 0 & -28.6 \\
2.2475 & 68.89 & -30.225 &  -0.981 & 0 & -35.085 & 380.38 \\
0 & -0.2708 & 0 & 0 & 1 & 0 & 0 \\
0.0271 & 0.2414 & -0.0154 & 0 & 0 & 0 & 7.7441 \\
-0.038 & -2.0037 & 0.0216 & 0.2656 & 0 & 0 & -45.6244 \\
0.2197 & 16.896 & -6.2749 & -0.3724 & 0 & 0 & 285.427 \\
\end{bmatrix}
\end{align}
is updated and the corresponding flag is set to zero, i.e.,
\begin{align}\label{eq:example_set_f2_zero}
f_2 \leftarrow 0.
\end{align}
It is noted that now $\vect{f} = \vect{0}$ and, thus, the iterative procedure is stopped. Furthermore, $\vect{Z}$ given in \eqref{eq:example_Z_final} has become an invertible $n \times n$ matrix and ${\nu_1 + \nu_2 = n}$. From ${\nu_2 = 4 > \nu_1 = 3}$ it follows that ${j_1 = 2}$ and ${j_2 = 1}$ and, thus,
\begin{align}\label{eq:example_step1_final_result}
\begin{split}
&\mu_1 = \nu_2 = 4\comma \quad \mu_2 = \nu_1 = 3 \comma \quad
\vect{\Gamma} =
\begin{bmatrix}
\vect{\psi}_{2,1}\T \\
\vect{\psi}_{1,1}\T \\
\end{bmatrix}
=
\begin{bmatrix}
-0.2708 & 1 \\
1 & 0
\end{bmatrix} \comma \quad
\vect{\Xi} = \vect{Z}^{-1}
\begin{bmatrix}
\vect{\psi}_{1,2}\T \\
\vect{\psi}_{1,3}\T \\
\vect{0}\T \\
\vect{\psi}_{2,2}\T \\
\vect{\psi}_{2,3}\T \\
\vect{\psi}_{2,4}\T \\
\vect{0}\T
\end{bmatrix}
=
\begin{bmatrix}
0 & 0 \\
0 & 0 \\
0 & 0 \\
0 & 0 \\
-2.006 & 0 \\
0 & 0 \\
0 & 0
\end{bmatrix}
.
\end{split}
\end{align}
Finally, the matrices of the auxiliary system are given by
\begin{align}\label{eq:example_matrices_of_auxiliary_system}
\begin{split}
&\vect{\check{A}} = \vect{A} + \vect{\Xi}\vect{C} =
\begin{bmatrix}
-0.3 & 0 & -33 & 9.81 & 0 & -5.4 & 0 \\
-0.1 & -8.3 & 3.75 & 0 & 0 & 0 & -28.6 \\
0.37 & 0 & -0.64 & 0 & 0 & -9.5 & 0 \\
0 & 1 & 0 & 0 & 0 & 0 & 0 \\
0 & -2.006 & 1 & 0 & 0 & 0 & 0 \\
0 & 0 & 0 & 0 & 0 & -10 & 0 \\
0 & 0 & 0 & 0 & 0 & 0 & -5
\end{bmatrix}\comma \quad
\vect{\check{C}} = \vect{\Gamma}\vect{C} =
\begin{bmatrix}
0 & 0 \\
-0.2708 & 1 \\
0 & 0 \\
0 & 0 \\
1 & 0 \\
0 & 0 \\
0 & 0 \\
\end{bmatrix}\T
.
\end{split}
\end{align}
\end{adjustwidth}

\subsubsection*{Step 2: Calculation of $p$ columns of the state transformation matrix}

\begin{adjustwidth}{1.5em}{0pt}
	Calculation of the matrix
	\begin{align}\label{eq:example_reduced_observability_matrix}
	\obsvred =
	\begin{bmatrix}
	\vect{\check{c}}_1\T \\
	\vect{\check{c}}_1\T \vect{\check{A}} \\
	\vect{\check{c}}_1\T \vect{\check{A}}^{2} \\
	\vect{\check{c}}_1\T \vect{\check{A}}^{3} \\
	\vect{\check{c}}_2\T \\
	\vect{\check{c}}_2\T \vect{\check{A}} \\
	\vect{\check{c}}_2\T \vect{\check{A}}^2 \\
	\end{bmatrix}
	=
	\begin{bmatrix}
	0 & -0.2708 & 0 & 0 & 1 & 0 & 0 \\
	0.0271 & 0.2414 & -0.0154 & 0 & 0 & 0 & 7.7441 \\
	-0.038 & -2.0037 & 0.0216 & 0.2656 & 0 & 0 & -45.6244 \\
	0.2197 & 16.896 & -6.2749 & -0.3724 & 0 & 0 & 285.427 \\
	0 & 1 & 0 & 0 & 0 & 0 & 0 \\
	-0.1 & -8.3 & 3.75 & 0 & 0 & 0 & -28.6 \\
	2.2475 & 68.89 & -30.225 & -0.981 & 0 & -35.085 & 380.38
	\end{bmatrix}
	\end{align}
	allows for determining the columns $\vect{t}_4$ and $\vect{t}_7$ of the state transformation matrix $\vect{T}$ as
	\begin{align}\label{eq:example_t_mu_j_major}
	\begin{bmatrix}
	\vect{t}_{4} & \vect{t}_{7}
	\end{bmatrix}
	=
	\obsvred^{-1}
	\begin{bmatrix}
	\vect{e}_{4} & \vect{e}_{7}
	\end{bmatrix} =
	\begin{bmatrix}
	-1.6452 & 0 \\
	0 & 0 \\
	0 & 0 \\
	0.7529 & 0 \\
	0 & 0 \\
	-0.0641 & -0.0285 \\
	0.0058 & 0
	\end{bmatrix}
	.
	\end{align}
	
\end{adjustwidth}

\subsubsection*{Step 3: Calculation of the coefficients $\beta_{j,k,l}$}

\begin{adjustwidth}{1.5em}{0pt}
	The coefficients $b_{1,2,l}$ are calculated in the following way:
	\begin{enumerate}
		\item In the given example, the matrix $\Abeta^{(1)}$ consists of one single submatrix which actually is a scalar and equal to one, i.e.,
		\begin{align}\label{eq:tutorial_H_beta_j}
		\Abeta^{(1)} =
		\begin{bmatrix}
		\Abeta^{(1)}_{2,2}
		\end{bmatrix}
		=
		1.
		\end{align}
		\item The corresponding vector $\vect{\bbeta}^{(1)}$ is also scalar in this case and is given by
		\begin{align}\label{eq:_tutorial_bbeta_j}
		\vect{\bbeta}^{(1)} =
		\begin{bmatrix}
		\vect{\bbeta}^{(1)}_{2}
		\end{bmatrix}
		=
		\vect{\check{c}}_2\T \vect{\check{A}}^{3}\vect{t}_4
		=
		-14.9336.
		\end{align}
		\item The coefficient $\beta_{1,2,3}$ is given by
		\begin{align}\label{eq:tutorial_beta_equation_sys}
		\vect{\beta}^{(1)} =
		\beta_{1,2,3}
		=
		\left( \Abeta^{(1)} \right)^{-1} \bbeta^{(1)}
		=
        -14.9336.
		\end{align}
		The remaining coefficients are set to zero, i.e., $\beta_{1,2,1} = \beta_{1,2,2} = 0$.
	\end{enumerate}
\end{adjustwidth}

\subsubsection*{Step 4: Construction of the state transformation matrix}

\begin{adjustwidth}{1.5em}{0pt}
	The remaining columns of $\vect{T}$ are given by
	\begin{align}\label{eq:tutorial_t_i_minor}
	\begin{split}
	&\vect{t}_1 = \vect{\check{A}}^3 \vect{t}_4 - \beta_{1,2,3} \vect{\check{A}}^2 \vect{t}_7 \text{,}\\
	&\vect{t}_2 = \vect{\check{A}}^2 \vect{t}_4 - \beta_{1,2,3} \vect{\check{A}} \vect{t}_7 \text{,}\\
	&\vect{t}_3 = \vect{\check{A}} \vect{t}_4 - \beta_{1,2,3} \vect{t}_7 \text{,}\\
	&\vect{t}_5 = \vect{\check{A}}^2 \vect{t}_7 \text{,}\\
	&\vect{t}_6 = \vect{\check{A}} \vect{t}_7 \text{,}
	\end{split}
	\end{align}
	which finally yields the state transformation matrix
	\begin{align}\label{eq:tutorial_T}
	\vect{T} =
	\begin{bmatrix}
	-20.2951 & -3.6294 & 8.2259 & -1.6452 & -10.5207 & 0.1539 & 0 \\
	0 & 0 & 0 & 0 & 1 & 0 & 0 \\
	18.4529 & 1 & 0 & 0 & -2.8241 & 0.2708 & 0 \\
	0 & 0 & 0 & 0.7529 & 0 & 0 & 0 \\
	1 & 0 & 0 & 0 & 0.2708 & 0 & 0 \\
	21.5113 & -2.1511 & 0.2151 & -0.0641 & -2.8502 & 0.285 & -0.0285 \\
	-0.719 & 0.1438 & -0.0288 & 0.0058 & 0 & 0 & 0
	\end{bmatrix}.
	\end{align}
\end{adjustwidth}

Applying the state transformation $\vect{\bar{x}} = \vect{T}^{-1} \vect{x}$ and the output transformation $\vect{\bar{y}} = \vect{\Gamma}\vect{y}$ yields the system
\begin{align}\label{eq:tutorial_transformed_system}
\begin{split}
\overset{\vect{.}}{\vect{\bar{x}}} &= \vect{\bar{A}}\vect{\bar{x}} + \vect{\bar{B}}\vect{u} + \vect{\bar{D}}\pert\text{,}\\
\vect{\bar{y}} &= \vect{\bar{C}}\vect{\bar{x}}\comma
\end{split}
\end{align}
in the proposed observer normal form, where the matrices
\begin{align}\label{eq:tutorial_transformed_system_matrices}
\begin{split}
\vect{\bar{A}} =
\left[\begin{array}{cccc|ccc}
-6.4019 & 1 & 0 & 0 & 2.006 & 0 & 0 \\
-7.0093 & 0 & 1 & 0 & 11.1769 & 0 & 0 \\
0 & 0 & 0 & 1 & 6 & 0 & 0 \\
0 & 0 & 0 & 0 & 1.3281 & 0 & 0 \\
\hline
91.7924 & 0 & 0 & 0 & -17.8381 & 1 & 0 \\
593.4635 & 0 & 0 & 0 & -271.7324 & 0 & 1 \\
0 & 0 & 0 & -14.9336 & -1220.7850 & 0 & 0
\end{array}\right]
\text{,}\quad
\vect{\bar{D}} =
\begin{bmatrix}
0 \\
0 \\
0 \\
0 \\
\hline
0 \\
0 \\
-701.7
\end{bmatrix}
\text{,} \quad
\vect{\bar{C}} =
\begin{bmatrix}
1 & 0 \\
0 & 0 \\
0 & 0 \\
0 & 0 \\
\hline
0 & 1 \\
0 & 0 \\
0 & 0
\end{bmatrix}\T
\text{,}
\end{split}
\end{align}
are calculated according to \eqref{eq:matrices_of_transformed_system} and
\begin{align}\label{eq:example_Bbar}
\vect{\bar{B}} = \vect{T}^{-1}\vect{B} =
\begin{bmatrix}
	0 & 0 \\
	0 & 77.44050 \\
	0 & 39.5191 \\
	0 & 0 \\
	\hline
	0 & 0 \\
	0 & -286 \\
	-701.7 & -8406.3602
\end{bmatrix}.
\end{align}

\subsection{Design of a robust observer}\label{ssec:tutorial_robust_observer}
In order to reconstruct the state vector $\vect{\bar{x}}$ despite the unknown input $\pert$, a robust observer is designed according to Section \ref{ssec:robust_observer_design} which yields
\begin{align}\label{eq:tutorial_robust_observer_of_transformed_system}
\begin{split}
\overset{\vect{.}}{\vect{\hat{\bar{x}}}} &= \vect{\bar{A}}\vect{\hat{\bar{x}}} + \vect{\bar{B}}\vect{u} + \vect{\bar{\Pi}} \vect{\sigma}_{\vect{\bar{y}}} + \vect{\bar{l}}(\vect{\sigma}_{\vect{\bar{y}}})  \text{,}\\
\vect{\hat{\bar{y}}} &= \vect{\bar{C}}\vect{\hat{\bar{x}}}\comma
\end{split}
\end{align}
where
\begin{align}\label{eq:tutorial_output_error}
\vect{\sigma}_{\vect{\bar{y}}}  = \vect{\bar{y}} - \vect{\hat{\bar{y}}} =
\begin{bmatrix}
\sigma_1 & \sigma_5
\end{bmatrix}\T
\end{align}
is the output error,
\begin{align}\label{eq:tutorial_alpha_j_1}
\vect{\bar{\Pi}} =
\begin{bmatrix}
-6.4019 & 2.006 \\
-7.0093 & 11.1769 \\
0 & 6 \\
0 & 1.3281 \\
\hline
91.7924 & -17.8381 \\
593.4635 & -271.7324 \\
0 & -1220.7850
\end{bmatrix}
\end{align}
is the linear output injection matrix and
\begin{align}\label{eq:tutorial_nonlinear_output_injection}
\begin{split}
\vect{\bar{l}}(\vect{\sigma}_{\vect{\bar{y}}}) =
\left [ \begin{array}{cccc|ccc}
\kappa_{1,3}\krax{\sigma_1}{\frac{3}{4}} & \kappa_{1,2}\krax{\sigma_1}{\frac{1}{2}} & \kappa_{1,1}\krax{\sigma_1}{\frac{1}{4}} & \kappa_{1,0}\krax{\sigma_1}{0} \hspace*{0.2cm} & \hspace*{0.2cm} \kappa_{2,2}\krax{\sigma_5}{\frac{2}{3}} & \kappa_{2,1}\krax{\sigma_5}{\frac{1}{3}} & \kappa_{2,0}\krax{\sigma_5}{0}
\end{array} \right ]\T
\end{split}
\end{align}
is the nonlinear output injection vector. Note that the control input $\vect{u}$ is known and, thus, is taken into account by the observer \eqref{eq:tutorial_robust_observer_of_transformed_system}. Then, the resulting estimation error dynamics do not depend on $\vect{u}$.
Furthermore, since the unknown input $\pert$ does not directly act on the first subsystem, it is also possible to choose the nonlinear output injection of the first subsystem according to e.g. a non-robust but continuous finite-time differentiator \citep{CRUZZAVALA2016660}, i.e.,
\begin{align}\label{eq:tutorial_nonlinear_output_injection_continuous}
\begin{split}
\resizebox{.9\width}{!}{$
\vect{\bar{l}}(\vect{\sigma}_{\vect{\bar{y}}}) \curvearrowright
\vect{\bar{l}}'(\vect{\sigma}_{\vect{\bar{y}}}) =
\left [ \begin{array}{cccc|ccc}
\kappa'_{1,3}\krax{\sigma_1}{\frac{r_2}{r_1}} & \kappa'_{1,2}\krax{\sigma_1}{\frac{r_3}{r_1}} & \kappa_{1,1}'\krax{\sigma_1}{\frac{r_4}{r_1}} & \kappa_{1,0}'\krax{\sigma_1}{\frac{r_5}{r_1}} \hspace*{0.2cm} & \hspace*{0.2cm} \kappa_{2,2}\krax{\sigma_5}{\frac{2}{3}} & \kappa_{2,1}\krax{\sigma_5}{\frac{1}{3}} & \kappa_{2,0}\krax{\sigma_5}{0}
\end{array} \right ]\T\text{,}$}
\end{split}
\end{align}
where $r_i = 1 - (n-i)q$, $i = 1, \dots, n+1$ are the so-called dilation coefficients and the parameter $q \in (-1, 0)$ is the homogeneity degree to be chosen.

\subsection{Simulation}\label{ssec:tutorial_simulation}

In the simulation, the unknown input
\begin{align}\label{eq:tutorial_disturbance}
\pert = 0.008 + 0.01 \sin(2t) + 0.002 \cos(13t)
\end{align}
with an amplitude bounded by $L = 0.02$ is considered. Furthermore, no control signal is applied to the system, i.e., $\vect{u} = \vect{0}$, which results unbounded state variables due to the unstable eigenvalues of $\vect{A}$. The observer parameters
\begin{align}\label{eq:tutorial_kappa}
\begin{split}
&\kappa_{1,3} = 5.3348\text{,}\quad \kappa_{1,2} = 13.1556\text{,}\quad \kappa_{1,1} = 17.2047\text{,}\quad \kappa_{1,0} = 11\text{,}\\
& \kappa_{2,2} = 9.6484\text{,}\quad \kappa_{2,1} = 49.369\text{,}\quad \kappa_{2,0} = 123.4992\text{,}
\end{split}
\end{align}
of the nonlinear output injection \eqref{eq:tutorial_nonlinear_output_injection} are chosen according to \citep[Section 6.7]{Shtessel2013}. Note that the necessary condition ${\kappa_{2,0} > L \abs{\bar{d}_{7,1}} = 14.0340}$ given in Theorem \ref{THEOREM:OBSERVER} is satisfied.

The initial state vector of the plant is selected as ${\vect{x}(0) =\begin{bmatrix} -0.5 & 0.1 & 0.02 & 0.2 & -0.1 & -0.3 & 0.2 \end{bmatrix}\T}$ and the initial observer state vector is set to ${\vect{\hat{\bar{x}}}(0) = \vect{0}}$. The estimation error variables $\eta_i = x_i - \hat{x}_i$, $i = 1, \dots, 7$ obtained by a numerical simulation are shown in Figure \ref{fig:lateral_motion_aircraft}.
\begin{figure}[b]
	\begin{center} \vspace{0.4cm}
        \includegraphics{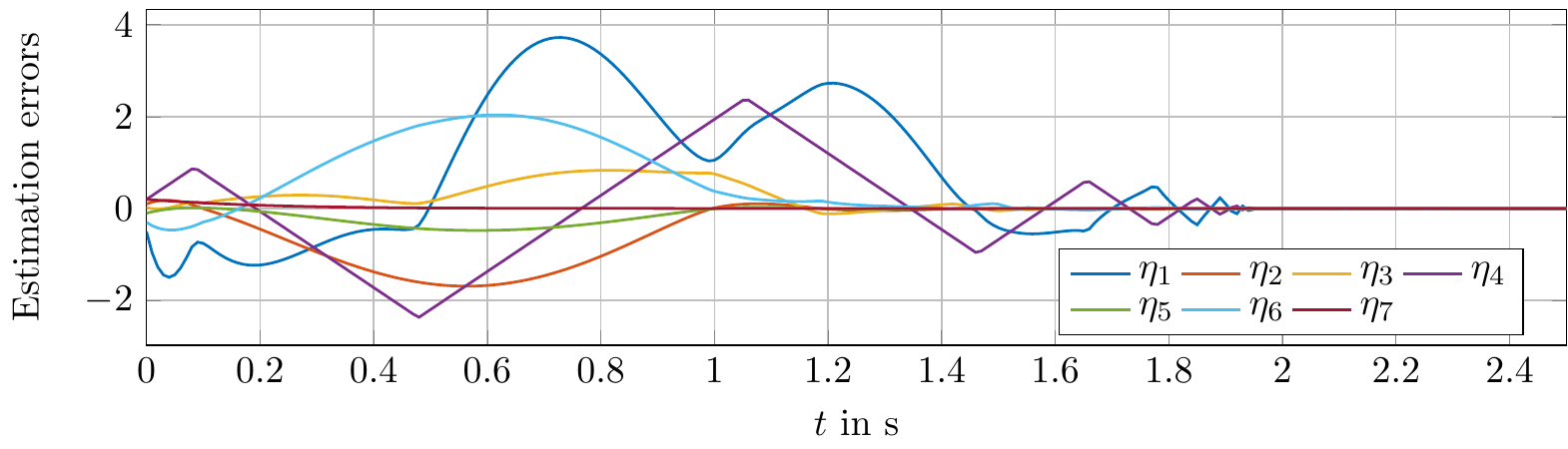}
		\caption{Estimation error variables $\eta_i = x_i - \hat{x}_i$, $i = 1, \dots, 7$. The estimation error variables vanish within finite-time despite the unknown input $\pert$.}\label{fig:lateral_motion_aircraft} 
	\end{center}
\end{figure}
After a finite convergence time of approximately $2$ s, the error variables vanish despite the unknown input $\pert$ and, thus, the state variables are exactly reconstructed.

\subsection{Reconstruction of the unknown input and robust state-feedback control}\label{ssec:tutorial_reconstruction_state_feedback}

For fault detection applications and robust state-feedback control, the knowledge of the unknown input $\pert$ in addition to the state variables is highly beneficial. The method for the robust observer design proposed in this article can be exploited in order to reconstruct the unknown input $\pert$. For this purpose, the state vector is augmented by $\pert$, i.e.,
\begin{align}\label{eq:tutorial_augmented_state_vector}
\vect{\chi} = 
\begin{bmatrix}
\vect{x}\T & \pert
\end{bmatrix}\T.
\end{align}
Rewriting the plant in terms of the augmented state vector yields
\begin{align}\label{eq:tutorial_augmented_system}
\begin{split}
\dot{\vect{\chi}} &=  \vect{A}_\chi \vect{\chi} + \vect{B}_\chi \vect{u} + \vect{D}_\chi \dot{\pert}\text{,}\\
\vect{y} &= \vect{C}_\chi\vect{\chi}\text{,}
\end{split}
\end{align}
where
\begin{align}\label{eq:tutorial_augmented_system_matrices}
\begin{split}
\vect{A}_\chi =
\begin{bmatrix}
\vect{A} & \vect{D} \\
\vect{0}\T & 0
\end{bmatrix}\text{,} \quad
\vect{B}_\chi =
\begin{bmatrix}
\vect{B} \\ \vect{0}\T
\end{bmatrix}\text{,} \quad
\vect{D} _\chi =
\begin{bmatrix}
\vect{0} \\ 1
\end{bmatrix}\text{,} \quad
\vect{C}_\chi =
\begin{bmatrix}
\vect{C} & \vect{0}
\end{bmatrix}.
\end{split}
\end{align}
Since the original system \eqref{eq:tutorial_original_system} is strongly observable, also the augmented system \eqref{eq:tutorial_augmented_system} with the unknown input $\dot{\pert}$ is strongly observable and, thus, the augmented system $\eqref{eq:tutorial_augmented_system}$ can be transformed into the proposed observer normal form. The application of the transformation algorithm yields the matrices of the corresponding system in the proposed observer normal form given by
\begin{align}\label{eq:tutorial_augmented_transformed_system_matrices}
\begin{split}
&\vect{\bar{A}}_\chi =
\left[\begin{array}{cccc|cccc}
-17.8381 & 1 & 0 & 0 & 91.7924  & 0 & 0 & 0 \\
-271.7324 & 0 & 1 & 0 & 593.4635 & 0 & 0 & 0 \\
-1220.785 & 0 & 0 & 1 & 0 & 0 & 0 & 0 \\
-19.8337 & 0 & 0 & 0 & 0 & 0 & 0 & 0 \\
\hline
2.006 & 0 & 0 & 0 & -6.4019 & 1 & 0 & 0 \\
11.1769 & 0 & 0 & 0 & -7.0093 & 0 & 1 & 0 \\
6 & 0 & 0 & 0 & 0 & 0 & 0 & 1 \\
1.328 & 0 & 0 & 0 & 0 & 0 & 0 & 0
\end{array}\right] \text{,} \quad
\vect{\bar{D}} _\chi =
\begin{bmatrix}
0 \\
0 \\
0 \\
-701.7 \\
\hline
0 \\
0 \\
0 \\
0
\end{bmatrix}\text{,} \\
&\vect{\bar{B}}_\chi =
\begin{bmatrix}
0 & 0 \\
0 & -286 \\
-701.7 & -8406.3602 \\
0 & 0 \\
\hline
0 & 0 \\
0 & 77.4405 \\
0 & 39.5191 \\
0 & 0
\end{bmatrix}\text{,} \quad
\vect{\bar{C}}_\chi =
\left[\begin{array}{cccc|cccc}
1 & 0 & 0 & 0 & 0 & 0 & 0 & 0 \\
0 & 0 & 0 & 0 & 1 & 0 & 0 & 0
\end{array}\right].
\end{split}
\end{align}
Furthermore, the amplitude of $\dot{\pert}$ is bounded by $L_{\dot{\pert}} = 0.046$ which again allows for the construction of a robust observer for the system in the proposed observer normal form according to Section \ref{ssec:robust_observer_design}. In the simulation, the observer parameters
\begin{align}\label{eq:tutorial_kappa_augmented}
\begin{split}
&\kappa_{1,3} = 9.9838\text{,}\quad \kappa_{1,2} = 46.0741\text{,}\quad \kappa_{1,1} = 112.7629\text{,}\quad\ \kappa_{1,0} = 134.9229\text{,}\\
&\kappa_{2,3} = 5.3348\text{,}\quad \kappa_{2,2} = 13.1556\text{,}\quad \kappa_{2,1} = 17.2047\text{,}\quad \kappa_{2,0} = 11\text{,}
\end{split}
\end{align}
are chosen according to \citep[Section 6.7]{Shtessel2013}. Again, the necessary condition ${\kappa_{1,0} > L_{\dot{\pert}} \abs{\bar{d}_{\chi,4,1}} = 32.2782}$ given in Theorem \ref{THEOREM:OBSERVER} is satisfied. Since $\pert$ acts as a matched disturbance, i.e., $\vect{D} = \vect{B}\vect{h}$ with $\vect{h} = \begin{bmatrix} 1 & 0 \end{bmatrix}\T$, its estimate $\hat{\pert}$ can be incorporated in the control law in order to compensate for the effects of the unknown input and achieve a robust stabilization of the state variables. A linear state-feedback controller given by
\begin{align}\label{eq:tutorial_state_feedback_augmented}
\vect{u} =
-
\begin{bmatrix}
\vect{K} & \vect{h}
\end{bmatrix}
\vect{\hat{\chi}}
\end{align}
is applied, where the state-feedback matrix $\vect{K}$ is designed for the nominal plant with $\pert \equiv 0$ and ${\vect{\hat{\chi}} = \begin{bmatrix} \vect{\hat{x}}\T & \hat{\pert} \end{bmatrix}\T}$ contains the estimates of the state variables and the unknown input.
The particular choice
\begin{align}\label{eq:tutorial_K_feedback}
\vect{K} =
\begin{bmatrix}
-0.0235 &  0.0036 & -0.2040 &  0.0356 & -0.1392 &  0.2471 & -0.0689 \\
-0.0258 & -0.0749 &  0.8998 & -0.6702 & -0.9164 & -0.7995 & -0.1821
\end{bmatrix}\text{,}
\end{align}
keeps the stable plant eigenvalues and transfers the unstable plant eigenvalues $s_1 = 0$ and $s_2 = 0.1219$ of $\vect{A}$ to $s'_1 = -1$ and $s'_2 = -2$ in the nominal closed-loop system with the dynamic matrix $\vect{A}-\vect{B}\vect{K}$.
The initial state vector of the plant is again selected as ${\vect{x}(0) =\begin{bmatrix} -0.5 & 0.1 & 0.02 & 0.2 & -0.1 & -0.3 & 0.2 \end{bmatrix}\T}$ and the initial observer state vector is set to ${\vect{\hat{\chi}}(0) = \vect{0}}$. Figure~\ref{fig:lateral_motion_aircraft_disturbance_reconstruction} shows the state variables, the estimation error variables and the estimate of the unknown input $\pert$ obtained from numerical simulations.

\begin{figure}[hb]
	\begin{center} \vspace{0.4cm}
        \includegraphics{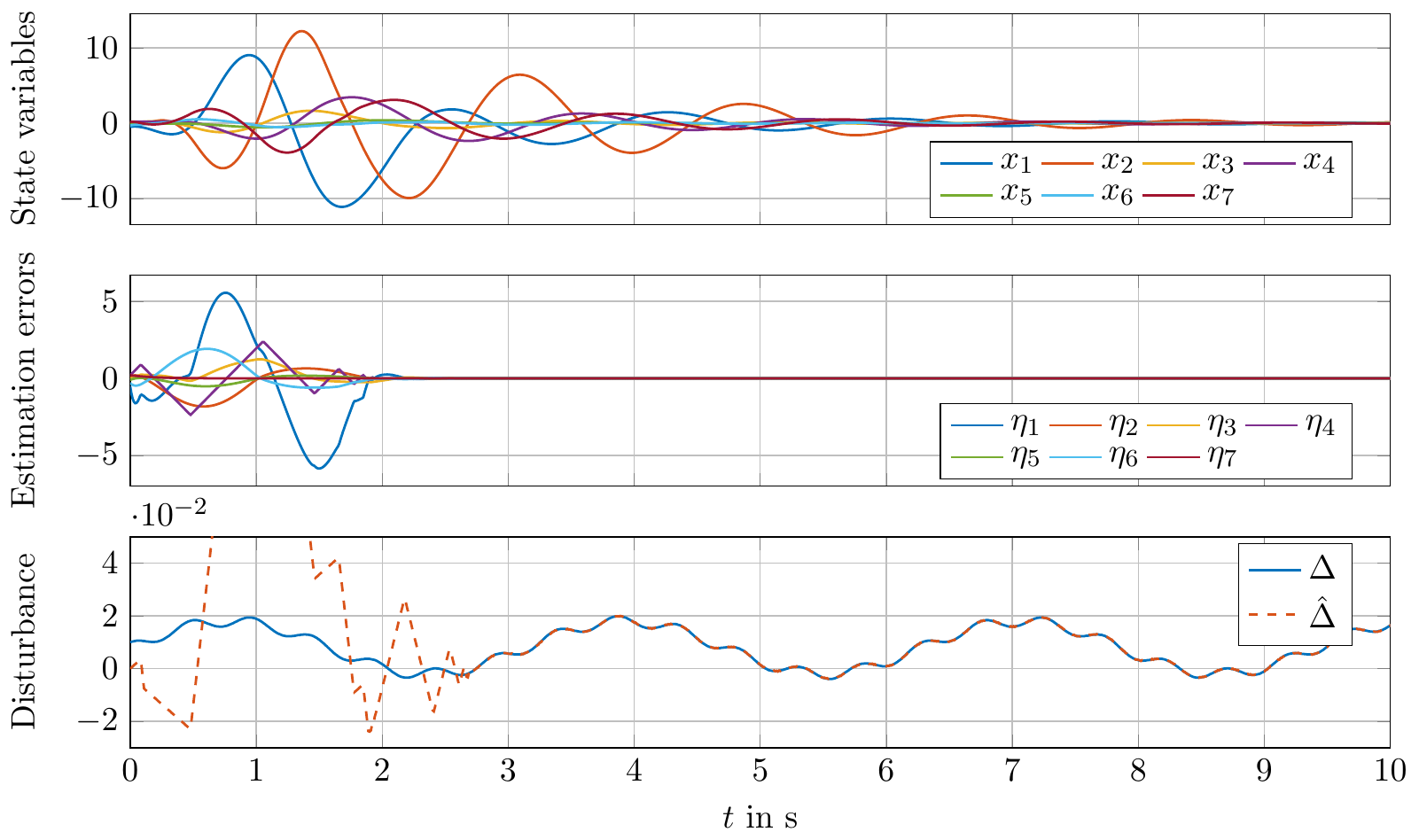}
		\caption{State variables $x_i$ of the plant, estimation errors $\eta_i = x_i - \hat{x}_i$ and estimate $\hat{\pert}$ of the unknown input $\pert$. The observer provides exact estimates of the plant states and the unknown input within a finite time and the states variables of the plant converge to zero asymptotically despite the unknown input.}\label{fig:lateral_motion_aircraft_disturbance_reconstruction} 
	\end{center}
\end{figure}

After a finite convergence time of approximately $2.5$ s the estimation errors vanish despite the unknown input $\pert$ and, thus, exact estimates of the state variables and the unknown input are obtained. Moreover, the state variables converge asymptotically to zero despite the unknown input once the observer has converged. Note that the state variables of the plant stay bounded, since the estimation errors stay bounded during convergence of the observer and the plant under consideration is linear and does not offer a finite escape time.

\section{Conclusion}\label{sec:Conclusion}

In this article, a novel observer normal form for strongly observable multivariable LTI systems with unknown inputs has been proposed. In contrast to classical approaches, this observer normal form explicitly considers the impact of the unknown inputs, which allows for a straight-forward construction of a RED-based observer. This observer provides theoretically exact estimates of the plant's state variables within finite time even in the presence of unknown bounded inputs. In contrast to many classical unknown input observers, the so-called observer matching condition is not prerequisite for the design of the observer. Moreover, the boundedness of the plant's state variables is also not required for global finite-time stability of the estimation error dynamics which allows for its application to unstable plants. In contrast to already existing sliding mode approaches %
which also do not require bounded state variables, the observer order and the number of tuning parameters are equal to the system order and are not unnecessarily increased.

The proposed method can be exploited for the reconstruction of unknown inputs with bounded derivatives which has been demonstrated by means of a simulation example. This approach in combination with a linear state-feedback controller allows for the robust stabilization of the plant in the presence of matched disturbances.

The proposed observer can be regarded as a multivariable generalization of the observer proposed in \citep{Niederwieser2019}. In the case of a single-output system the proposed observer normal form coincides with the well-known transposed observable canonical form for linear single-output systems. Furthermore, the proposed observer normal form allows to combine different observer concepts.

\section*{Acknowledgement}
The authors would like to thank Roland Falkensteiner for the constructive discussions and the important input during the coffee breaks.

\section*{Funding}

This work was supported by the European Union’s Horizon 2020 research and innovation programme under the Marie Skłodowska-Curie Grant agreement 734832.
The financial support by the Christian Doppler Research Association, Austria, the Austrian Federal Ministry for Digital and Economic Affairs and the National Foundation for Research, Technology and Development, Austria is gratefully acknowledged.

\appendix

\section{Systems with direct feed-through}\label{app:systems_with_direct_feedthrough}

In the following, the handling of systems with direct feed-through, which is already mentioned in Remark \ref{rem:system_with_direct_feedthrough}, is described in more detail.

\subsection{Reformulation as system without direct feed-through}\label{app:remove_direct_feedthrough}

In order to obtain an equivalent representation of system \eqref{eq:system_with_direct_feed_through} without direct feed-through, the algorithm proposed in \citep[Section 5.4, page 155]{Chen2004} is applied. Firstly, regular transformation matrices $\vect{U} \in \mathbb{R}^{p \times p}$ for the output and $\vect{V} \in \mathbb{R}^{m \times m}$ for the input are chosen such that
\begin{align}\label{eq:system_with_direct_feed_through_condition_on_transformation_matrices}
\begin{split}
\vect{U}\vect{F}\vect{V} = \begin{bmatrix}	\vect{I}_{m_\text{F}} & \vect{0} \\ \vect{0} & \vect{0} \end{bmatrix} \comma
\end{split}
\end{align}
where $m_\text{F} = \rank \vect{F}$.
Application of the output and input transformation
\begin{align}\label{eq:system_with_direct_feed_through_transformations_appendix}
\vect{\tilde{y}} = \vect{U}\vect{y} = \begin{bmatrix}	\vect{\tilde{y}}_0 \\ \vect{\tilde{y}}_1 \end{bmatrix}\comma \qquad
\vect{\tilde{\pert}} 	= \vect{V}^{-1}\vect{\pert} = \begin{bmatrix}	\vect{\tilde{\pert}}_0 \\ \vect{\tilde{\pert}}_1 \end{bmatrix}\comma
\end{align}
allows to represent system \eqref{eq:system_with_direct_feed_through} in terms of the new output $\vect{\tilde{y}}$ and the new input $\vect{\tilde{\pert}}$, i.e.,
\begin{subequations}\label{eq:system_with_direct_feed_through_transformed}
\begin{align}\label{eq:system_with_direct_feed_through_transformed_1}
&\dot{\vect{x}} = \vect{A}\vect{x} + \begin{bmatrix} \vect{\tilde{D}}_0 & \vect{\tilde{D}}_1  \end{bmatrix} \begin{bmatrix} \vect{\tilde{\pert}}_0 \\ \vect{\tilde{\pert}}_1  \end{bmatrix} \text{,}\\
\label{eq:system_with_direct_feed_through_transformed_2}
&\begin{bmatrix}	\vect{\tilde{y}}_0 \\ \vect{\tilde{y}}_1 \end{bmatrix} = \begin{bmatrix} \vect{\tilde{C}}_0 \\ \vect{\tilde{C}}_1  \end{bmatrix} \vect{x} + \begin{bmatrix}	\vect{I}_{m_\text{F}} & \vect{0} \\ \vect{0} & \vect{0} \end{bmatrix} \begin{bmatrix}	\vect{\tilde{\pert}}_0 \\ \vect{\tilde{\pert}}_1 \end{bmatrix}\comma
\end{align}
\end{subequations}
where $\vect{\tilde{D}}_0 \in \mathbb{R}^{n \times m_\text{F}}$, $\vect{\tilde{D}}_1 \in \mathbb{R}^{n \times (m-m_\text{F})}$, $\vect{\tilde{C}}_0 \in \mathbb{R}^{m_\text{F} \times n}$ and $\vect{\tilde{C}}_1 \in \mathbb{R}^{(m-m_\text{F}) \times n}$. The output equation \eqref{eq:system_with_direct_feed_through_transformed_2} is solved for
\begin{align}\label{eq:system_with_direct_feed_through_Delta_0}
\vect{\tilde{\Delta}}_0 = \vect{\tilde{y}}_0 - \vect{\tilde{C}}_0 \vect{x}
\end{align}
and substituted into \eqref{eq:system_with_direct_feed_through_transformed_1} which yields
\begin{align}\label{eq:system_with_direct_feed_through_Delta_0_inserted}
&\dot{\vect{x}} = \vect{A}\vect{x} + \vect{\tilde{D}}_0 \left(\vect{\tilde{y}}_0 - \vect{\tilde{C}}_0 \vect{x}\right) + \vect{\tilde{D}}_1 \vect{\tilde{\pert}}_1.
\end{align}
Combining \eqref{eq:system_with_direct_feed_through_Delta_0_inserted} and the remaining outputs $\vect{\tilde{y}}_1$ given in \eqref{eq:system_with_direct_feed_through_transformed_2} finally results in a system representation without direct feed-through
\begin{align}\label{eq:app_system_with_direct_feed_through_reduced}
\begin{split}
&\dot{\vect{x}} = \vect{\tilde{A}}\vect{x} + \vect{\tilde{D}}_1\vect{\tilde{\pert}}_1 + \vect{\tilde{D}}_0\vect{\tilde{y}}_0\text{,}\\
&\vect{\tilde{y}}_1 = \vect{\tilde{C}}_1\vect{x}\comma
\end{split}
\end{align}
where $\vect{\tilde{A}} = \vect{A} - \vect{\tilde{D}}_0\vect{\tilde{C}}_0$.

\subsection{Preservation of strong observability}\label{app:invariance_wrt_strong_observability}

Since the strong observability property is invariant w.r.t. to regular transformations of the output and the input, strong observability of system \eqref{eq:system_with_direct_feed_through_transformed} directly follows from the strong observability of the original system \eqref{eq:system_with_direct_feed_through}. In other words, the Rosenbrock matrix of system \eqref{eq:system_with_direct_feed_through_transformed} satisfies
\begin{align}\label{eq:strong_observability_condition_rosenbrock_after_transformations}
\rank
\begin{bmatrix}
s \vect{I}_n - \vect{A} & -\vect{\tilde{D}}_0 & -\vect{\tilde{D}}_1 \\
\vect{\tilde{C}}_0 & \vect{I}_{m_\text{F}} & \vect{0} \\
\vect{\tilde{C}}_1 & \vect{0} & \vect{0}
\end{bmatrix}
= n + m \qquad \forall s \in \mathbb{C}.
\end{align}
Furthermore, consider the Rosenbrock matrix
\begin{align}\label{eq:rosenbrock_tilde}
\vect{\tilde{P}}(s) =
\begin{bmatrix}
s \vect{I}_n - (\vect{A} - \vect{\tilde{D}}_0\vect{\tilde{C}}_0) & -\vect{\tilde{D}}_1 \\
\vect{\tilde{C}}_1 & \vect{0}
\end{bmatrix}
\end{align}
of the system \eqref{eq:app_system_with_direct_feed_through_reduced} with the known input $\vect{\tilde{y}}_0$.
It directly follows that
\begin{align}\label{eq:strong_observability_condition_rosenbrock_tilde}
\begin{split}
\rank\vect{\tilde{P}}(s) &=
\rank \begin{bmatrix}
s \vect{I}_n - (\vect{A} - \vect{\tilde{D}}_0\vect{\tilde{C}}_0) & -\vect{\tilde{D}}_1 \\
\vect{0} & \vect{0} \\
\vect{\tilde{C}}_1 & \vect{0}
\end{bmatrix}
=\\
&=
\rank
\begin{bmatrix}
s \vect{I}_n - \vect{A} & -\vect{\tilde{D}}_0 & -\vect{\tilde{D}}_1 \\
\vect{\tilde{C}}_0 & \vect{I}_{m_\text{F}} & \vect{0} \\
\vect{\tilde{C}}_1 & \vect{0} & \vect{0}
\end{bmatrix}
\begin{bmatrix}
\vect{I}_n & \vect{0} \\
-\vect{\tilde{C}}_0 &  \vect{0} \\
\vect{0} & \vect{I}_{m-m_\text{F}}
\end{bmatrix} = \\
&\hspace*{-0.2cm}\overset{\eqref{eq:strong_observability_condition_rosenbrock_after_transformations}}{=} \rank \begin{bmatrix}
\vect{I}_n & \vect{0} \\
-\vect{\tilde{C}}_0 &  \vect{0} \\
\vect{0} & \vect{I}_{m-m_\text{F}}
\end{bmatrix} =
n+ m - m_\text{F} \qquad \forall s \in \mathbb{C}\comma
\end{split}
\end{align}
and, thus, system \eqref{eq:app_system_with_direct_feed_through_reduced} is strongly observable \citep{HAUTUS1983353}.

\section{Proof of Lemma \ref{LEM:TRANSFORMATION_LEMMA}}\label{app:transformation_lemma}

In the following, the parts a) to f) of Lemma \ref{LEM:TRANSFORMATION_LEMMA} are proven.

\subsection{Proof of Lemma \ref{LEM:TRANSFORMATION_LEMMA}a)}\label{app:properties_of_auxiliary_system}

Since Step 1 relies on the steps SCB.1, SCB.2 and SCB.3 of the decomposition algorithm proposed in \citep[Section 5.3, pages 119-127]{Chen2004} this proof mainly refers to the results given in the cited book. Therein, a structural decomposition of a general LTI system of the form \eqref{eq:original_system} into four parts is presented. In the context of this article the strongly observable parts (labelled as b and d in \citep{Chen2004}) need to be considered. The regular output transformation \eqref{eq:output_transformation} and the regular state transformation {$\vect{\tilde{x}} = \begin{bmatrix} \vect{\tilde{x}}_1\T & \dots & \vect{\tilde{x}}_p\T \end{bmatrix}\T = \vect{Z}\vect{x}$} decompose the system into $p$ coupled chains of integrators of the form
\begin{align}\label{eq:integrator_chain_single_output}
\begin{split}
&\overset{\vect{.}}{\vect{\tilde{x}}}_j =
\begin{bmatrix}
0      & 1      & 0      & \dots  & 0      \\
\vdots & \ddots & \ddots & \ddots & \vdots \\
\vdots &        & \ddots & \ddots & 0      \\
0      & \dots  & \dots  & 0      & 1      \\
0      & \dots  & \dots  & \dots  & 0
\end{bmatrix}
\vect{\tilde{x}}_j
+
\begin{bmatrix}
    \vect{0}\T \\ \vdots \\ \vdots \\ \vect{0}\T \\ \vect{\tilde{a}}_j\T
\end{bmatrix}
\vect{\tilde{x}}
-
\vect{\tilde{\Xi}}_j\vect{y}
+
\begin{bmatrix}
\vect{0}\T \\ \vdots \\ \vdots \\ \vect{0}\T \\ \vect{\tilde{d}}_j\T
\end{bmatrix}
\vect{\pert}\comma\\
&\bar{y}_j =
\begin{bmatrix}
1 & 0 & \dots & 0
\end{bmatrix}
\vect{\tilde{x}}_j\comma
\end{split}
\end{align}
with one single output $\bar{y}_j$, some linear combination of $\vect{\pert}$ as unknown input and some couplings $\vect{\tilde{a}}_j\T\vect{\tilde{x}}$ in the last differential equation. In this representation, $\vect{\pert}$ either occurs explicitly only in the $\mu_j^\text{th}$ derivative of $\bar{y}_j$ or does not explicitly act on the $j\text{th}$ subsystem at all if $\vect{\tilde{d}}_j\T = \vect{0}\T$, i.e., under the output injection $\vect{\tilde{\Xi}}_j\vect{y}$ the relative degree of $\bar{y}_j$ w.r.t. $\vect{\pert}$, if it exists, is equal or larger than $\mu_j$. Thus, by applying the output injection $\vect{\Xi} \vect{y} = \vect{Z}^{-1}\begin{bmatrix} \vect{\tilde{\Xi}}_1\T & \dots & \vect{\tilde{\Xi}}_p\T \end{bmatrix}\T \vect{y}$ and the output transformation \eqref{eq:output_transformation} to the original system \eqref{eq:original_system} one exactly obtains the auxiliary system \eqref{eq:auxiliary_system_step_1} satisfying the conditions \eqref{eq:rank_of_reduced_observability_matrix} and \eqref{eq:relative_degree_condition}.

\subsection{Proof of Lemma \ref{LEM:TRANSFORMATION_LEMMA}b)}\label{app1.1a}

In order to show that the transformation matrix $\vect{T}$ given in \eqref{eq:transformation_matrix} is non-singular, a procedure similar to the one in \citep{Luenberger1967} is applied. Consider constants $\xi_{j,i}$ such that the linear combination of the column vectors of the transformation matrix $\vect{T}$ equals the zero vector, i.e.,
\begin{align}\label{eq:linear_combination_t_i_1}
\sum\limits_{j=1}^p \sum_{i=0}^{\mu_j-1} \xi_{j,i} \vect{t}_{\mu_1+ \dots + \mu_j - i} = \vect{0}.
\end{align}
Substituting $\vect{t}_{\mu_1+ \dots + \mu_j - i}$ according to \eqref{eq:t_i} and multiplication of both sides with $\vect{c}\T_{\bar{j}} \vect{\check{A}}^{\bar{i}}$ from the left-hand side yields
\begin{align}\label{eq:linear_combination_t_i_2}
\sum\limits_{j=1}^p \sum_{i=0}^{\mu_j-1} \xi_{j,i} \left(\vect{\check{c}}\T_{\bar{j}}\vect{\check{A}}^{i+\bar{i}}\vect{t}_{\mu_1 + \dots + \mu_j} - \sum\limits_{r = j+1}^{p}\sum\limits_{q=1}^{i} \beta_{j,r,\mu_j - q} \vect{\check{c}}\T_{\bar{j}}\vect{\check{A}}^{i-q+\bar{i}}\vect{t}_{\mu_1+ \dots +\mu_r}\right) = 0.
\end{align}
Due to the choice of the vectors $\vect{t}_{\mu_1},\vect{t}_{\mu_1 + \mu_2},\dots,\vect{t}_n$ in \eqref{eq:t_mu_j_major}, they satisfy
\begin{align}\label{eq:linear_combination_t_i_3}
\begin{bmatrix}
\vect{\check{c}}_1\T \\
\vdots \\
\vect{\check{c}}_1\T \vect{\check{A}}^{\mu_1 - 1} \\
\vect{\check{c}}_2\T \\
\vdots \\
\vect{\check{c}}_2\T \vect{\check{A}}^{\mu_2 - 1} \\
\vdots \\
\vect{\check{c}}_p\T \\
\vdots \\
\vect{\check{c}}_p\T \vect{\check{A}}^{\mu_p - 1} \\
\end{bmatrix}
\begin{bmatrix}
\vect{t}_{\mu_1} & \vect{t}_{\mu_1 + \mu_2} & \dots & \vect{t}_{n}
\end{bmatrix}
=
\begin{bmatrix}
0      & 0      & \dots & 0 \\
\vdots & \vdots &       & \vdots \\
1      & 0      & \dots & 0 \\
0      & 0      & \dots & 0 \\
\vdots & \vdots &       & \vdots \\
0      & 1      & \dots & 0 \\
\vdots & \vdots &       & \vdots \\
0      & 0      & \dots & 0 \\
\vdots & \vdots &       & \vdots \\
0      & 0      & \dots & 1 \\
\end{bmatrix}.
\end{align}
Consider equation \eqref{eq:linear_combination_t_i_2} for $\bar{i} = 0$. Taking into account \eqref{eq:linear_combination_t_i_3} reduces \eqref{eq:linear_combination_t_i_2} to
\begin{align}\label{eq:linear_combination_t_i_4}
\xi_{\bar{j},\mu_{\bar{j}}-1} = 0 \qquad \text{for} \quad \bar{j} = 1,2,\dots,p.
\end{align}
Furthermore, consider equation \eqref{eq:linear_combination_t_i_2} for $\bar{i} = 1$. Taking into account \eqref{eq:linear_combination_t_i_3} and \eqref{eq:linear_combination_t_i_4} reduces \eqref{eq:linear_combination_t_i_2} to
\begin{align}\label{eq:linear_combination_t_i_5}
\xi_{\bar{j},\mu_{\bar{j}}-2} = 0 \qquad \text{for} \quad \bar{j} = 1,2,\dots,p.
\end{align}
Continuing in this way, $\xi_{j,i} = 0$ can be shown for $j = 1,2,\dots,p$, $i = 0,1,\dots,\mu_j - 1$. Hence, it can be concluded that the columns of the transformation matrix $\vect{T}$ are linearly independent and, thus, $\vect{T}$ is non-singular. Note that this proof does not require any restrictions regarding the coefficients $\beta_{j,k,l}$ and, thus, $\vect{T}$ is non-singular regardless of the values of $\beta_{j,k,l}$.%

\subsection{Proof of Lemma \ref{LEM:TRANSFORMATION_LEMMA}c)}\label{app1.1d}

The existence of a unique solution $\vect{\beta}^{(j)}$ of the system of equations \eqref{eq:beta_equation_sys} is ensured, if and only if the matrices $\vect{\Abeta}^{(j)}$ are non-singular in any case for $j = 1,\dots,p-1$. Consider the determinant
\begin{align}
\begin{split}
\det \vect{\Abeta}^{(j)} \overset{\eqref{eq:H_beta_j}}{=}
\begin{vmatrix}
\Abeta^{(j)}_{j+1,j+1} & \cdots & \Abeta^{(j)}_{j+1,p-1} & \Abeta^{(j)}_{j+1,p} \\
\vdots & \ddots & \vdots & \vdots \\
\Abeta^{(j)}_{p-1,j+1} & \cdots & \Abeta^{(j)}_{p-1,p-1} & \Abeta^{(j)}_{p-1,p} \\
\Abeta^{(j)}_{p,j+1} & \cdots & \Abeta^{(j)}_{p,p-1} & \Abeta^{(j)}_{p,p} \\
\end{vmatrix}.
\label{eq:proof_Abeta_1}
\end{split}
\end{align}
Since the block $\vect{\Abeta}^{(j)}_{p,p}$ of size $(\mu_j - \mu_p) \times (\mu_j - \mu_p)$ is an upper unitriangular matrix, see \eqref{eq:H_beta_j_k_l}, its determinant satisfies
\begin{align}
\begin{split}
\det\Abeta^{(j)}_{p,p} = 1.
\label{eq:proof_Abeta_3}
\end{split}
\end{align}
Furthermore, its inverse always exists and is again upper unitriangular, i.e.,
\begin{align}
\begin{split}
{\Abeta^{(j)}_{p,p}}^{-1} =
\begin{bmatrix}
1 & * & \cdots & * \\
0 & \ddots & \ddots & \vdots \\
\vdots & \ddots &\ddots & * \\
0 & \cdots & 0 & 1
\end{bmatrix}
\text{,}
\label{eq:proof_Abeta_1a}
\end{split}
\end{align}
which allows to rewrite the determinant \eqref{eq:proof_Abeta_1} in terms of the Schur complement of $\vect{\Abeta}^{(j)}_{p,p}$, i.e.,
\begin{align}
\begin{split}
\det \vect{\Abeta}^{(j)} &=
\underbrace{\det\Abeta^{(j)}_{p,p}}_{\overset{\eqref{eq:proof_Abeta_3}}{=} 1} \cdot 
\begin{vmatrix}
\begin{bmatrix}
\Abeta^{(j)}_{j+1,j+1} & \cdots & \Abeta^{(j)}_{j+1,p-1} \\
\vdots & \ddots & \vdots \\
\Abeta^{(j)}_{p-1,j+1} & \cdots & \Abeta^{(j)}_{p-1,p-1}
\end{bmatrix}
-
\begin{bmatrix}
\Abeta^{(j)}_{j+1,p} \\
\vdots \\
\Abeta^{(j)}_{p-1,p}
\end{bmatrix}
{\Abeta^{(j)}_{p,p}}^{-1}
\begin{bmatrix}
\Abeta^{(j)}_{p,j+1} & \cdots & \Abeta^{(j)}_{p,p-1}
\end{bmatrix}
\end{vmatrix} = \\
&=
\setlength\arraycolsep{10pt}
\renewcommand*{\arraystretch}{1.2} 
\begin{vmatrix}
\Abeta^{(j)}_{j+1,j+1} - \Abeta^{(j)}_{j+1,p}{\Abeta^{(j)}_{p,p}}^{-1}\Abeta^{(j)}_{p,j+1} & \cdots & \Abeta^{(j)}_{j+1,p-1} - \Abeta^{(j)}_{j+1,p}{\Abeta^{(j)}_{p,p}}^{-1}\Abeta^{(j)}_{p,p-1}\\
\vdots & \ddots & \vdots \\
\Abeta^{(j)}_{p-1,j+1} - \Abeta^{(j)}_{p-1,p}{\Abeta^{(j)}_{p,p}}^{-1}\Abeta^{(j)}_{p,j+1} & \cdots & \Abeta^{(j)}_{p-1,p-1} - \Abeta^{(j)}_{p-1,p}{\Abeta^{(j)}_{p,p}}^{-1}\Abeta^{(j)}_{p,p-1}
\end{vmatrix}.
\label{eq:proof_Abeta_2}
\end{split}
\end{align}
Furthermore, due to the special structure of the blocks $\Abeta^{(j)}_{r,p}$, $\Abeta^{(j)}_{p,s}$ given in \eqref{eq:H_beta_j_k_l} and ${\Abeta^{(j)}_{p,p}}^{-1}$ indicated in \eqref{eq:proof_Abeta_1a}, the blocks
\begin{align}
\Abeta^{(j,1)}_{r,s} = \Abeta^{(j)}_{r,s} - \Abeta^{(j)}_{r,p}{\Abeta^{(j)}_{p,p}}^{-1}\Abeta^{(j)}_{p,s} \qquad \text{for} \quad r = j+1,\dots,p-1\text{,}\quad s = j+1,\dots,p-1\text{,}
\label{eq:proof_Abeta_2b}
\end{align}
of the matrix in \eqref{eq:proof_Abeta_2} offer exactly the same structure as the original blocks $\Abeta^{(j)}_{r,s}$, i.e., the zero elements and also the ones on the main diagonal for the case $r = s$ of the original matrix $\Abeta^{(j)}$ are preserved. Thus, it is again possible to rewrite the determinant \eqref{eq:proof_Abeta_3} in terms of the Schur complement of $\Abeta^{(j,1)}_{p-1,p-1}$ which is again upper unitriangular, i.e.,
\begin{align}
\begin{split}
\det \vect{\Abeta}^{(j)} &=
\resizebox{.86\width}{!}{$
\underbrace{\det\Abeta^{(j,1)}_{p-1,p-1}}_{= 1} \cdot 
\begin{vmatrix}
\begin{bmatrix}
\Abeta^{(j,1)}_{j+1,j+1} & \cdots & \Abeta^{(j,1)}_{j+1,p-2} \\
\vdots & \ddots & \vdots \\
\Abeta^{(j,1)}_{p-2,j+1} & \cdots & \Abeta^{(j,1)}_{p-2,p-2}
\end{bmatrix}
-
\begin{bmatrix}
\Abeta^{(j,1)}_{j+1,p-1} \\
\vdots \\
\Abeta^{(j,1)}_{p-2,p-1}
\end{bmatrix}
{\Abeta^{(j,1)}_{p-1,p-1}}^{-1}
\begin{bmatrix}
\Abeta^{(j,1)}_{p-1,j+1} & \cdots & \Abeta^{(j,1)}_{p-1,p-2}
\end{bmatrix}
\end{vmatrix} =$} \\
&=
\setlength\arraycolsep{10pt}
\renewcommand*{\arraystretch}{1.2} 
\resizebox{.86\width}{!}{$\begin{vmatrix}
\Abeta^{(j,1)}_{j+1,j+1} - \Abeta^{(j,1)}_{j+1,p-1}{\Abeta^{(j,1)}_{p-1,p-1}}^{-1}\Abeta^{(j,1)}_{p-1,j+1} & \cdots & \Abeta^{(j,1)}_{j+1,p-2} - \Abeta^{(j,1)}_{j+1,p-1}{\Abeta^{(j,1)}_{p-1,p-1}}^{-1}\Abeta^{(j,1)}_{p-1,p-2}\\
\vdots & \ddots & \vdots \\
\Abeta^{(j,1)}_{p-2,j+1} - \Abeta^{(j,1)}_{p-2,p-1}{\Abeta^{(j,1)}_{p-1,p-1}}^{-1}\Abeta^{(j,1)}_{p-1,j+1} & \cdots & \Abeta^{(j,1)}_{p-2,p-2} - \Abeta^{(j,1)}_{p-2,p-1}{\Abeta^{(j,1)}_{p-1,p-1}}^{-1}\Abeta^{(j,1)}_{p-1,p-2}
\end{vmatrix}$}
.
\label{eq:proof_Abeta_4}
\end{split}
\end{align}
Again, the blocks
\begin{align}
\Abeta^{(j,2)}_{r,s} = \Abeta^{(j,1)}_{r,s} - \Abeta^{(j,1)}_{r,p-2}{\Abeta^{(j,1)}_{p-1,p-1}}^{-1}\Abeta^{(j,1)}_{p-2,s} \qquad \text{for} \quad r = j+1,\dots,p-2\text{,}\quad s = j+1,\dots,p-2\text{,}
\label{eq:proof_Abeta_5}
\end{align}
of the matrix in \eqref{eq:proof_Abeta_4} offer the same structure as $\Abeta^{(j,1)}_{r,s}$ and, thus, also as the original blocks $\Abeta^{(j)}_{r,s}$, which again allows to rewrite the determinant in \eqref{eq:proof_Abeta_4} in terms of the Schur complement of the upper unitriangular matrix $\Abeta^{(j,2)}_{p-2,p-2}$ and so on. Continuing this way until
\begin{align}
\begin{split}
\det \vect{\Abeta}^{(j)} &=
\underbrace{\det\Abeta^{(j,p-j-2)}_{j+2,j+2}}_{= 1} \cdot 
\begin{vmatrix}
\Abeta^{(j,p-j-2)}_{j+1,j+1}
-
\Abeta^{(j,p-j-2)}_{j+1,j+2}
{\Abeta^{(j,p-j-2)}_{j+2,j+2}}^{-1}
\Abeta^{(j,p-j-2)}_{j+2,j+1}
\end{vmatrix} = \\
&=
\det \Abeta^{(j,p-j-1)}_{j+1,j+1} = 1 \neq 0 \text{,}
\label{eq:proof_Abeta_6}
\end{split}
\end{align}
it finally follows that $\vect{\Abeta}^{(j)}$ is non-singular and, thus, there always exists a unique solution of the system of equations \eqref{eq:beta_equation_sys}.

\subsection{Proof of Lemma \ref{LEM:TRANSFORMATION_LEMMA}d)}\label{app1.1b}

Multiplication of \eqref{eq:matrices_of_transformed_system_A} with $\vect{T}$ from the left-hand side and substitution of $\vect{A}$ by the decomposition \eqref{eq:decompose_dynamic_matrix} yields
\begin{align}\label{eq:conditions_on_T_A_1}
\vect{T}\vect{\bar{A}} = \vect{\check{A}}\vect{T} - \vect{\Xi}\vect{C}\vect{T}.
\end{align}
The results of Lemma \ref{LEM:TRANSFORMATION_LEMMA}f) and the insertion of the identity matrix $\vect{\Gamma}^{-1}\vect{\Gamma}$ allow to simplify the term $\vect{\Xi}\vect{C}\vect{T}$ to
\begin{align}\label{eq:simplify_Xi_times_C}
\begin{split}
\vect{\Xi}\vect{C}\vect{T} &= \vect{\Xi}\vect{\Gamma}^{-1}\underbrace{\vect{\Gamma}\vect{C}\vect{T}}_{\overset{\eqref{eq:matrices_of_transformed_system_C}}{=} \vect{\bar{C}}} \overset{\eqref{eq:transformed_output_matrix}}{=}
\vect{\Xi}\vect{\Gamma}^{-1}
\begin{bmatrix}
\vect{e}_1\T \\
\vect{e}_{\mu_1 + 1}\T \\
\vdots \\
\vect{e}_{\mu_1 + \dots + \mu_{p-1}+ 1}\T
\end{bmatrix}
= \\[10pt]
&=
\left [ \begin{array}{cc|cc|c|cc}
\vect{\tilde{\alpha}}_1 & \vect{0}_{n \times (\mu_1 -1)} & \vect{\tilde{\alpha}}_2 & \vect{0}_{n \times (\mu_2 -1)} & \dots & \vect{\tilde{\alpha}}_p & \vect{0}_{n \times (\mu_p -1)}
\end{array} \right ]
\comma
\end{split}
\end{align}
where $\vect{\tilde{\alpha}}_j \in \mathbb{R}^{n}$, $j = 1,\dots,p$.
Using \eqref{eq:simplify_Xi_times_C} and the definitions of $\vect{\bar{A}}$ in \eqref{eq:transformed_dynamic_matrix} and $\vect{T}$ in \eqref{eq:transformation_matrix} allows to express equation \eqref{eq:conditions_on_T_A_1} in terms of its column vectors, i.e.,
\begin{align}\label{eq:conditions_on_T_A}
\begin{split}
\bigg [
&\resizebox{.83\width}{!}{$
	\begin{array}{cccc|}
	\vect{T} \vect{\alpha}_1 & \vect{t}_1 + \sum\limits_{r=2}^p\beta_{1,r,1} \vect{t}_{\mu_1+ \dots +\mu_r} & \dots & \vect{t}_{\mu_1-1} + \sum\limits_{r=2}^p\beta_{1,r,\mu_1-1} \vect{t}_{\mu_1+ \dots +\mu_r} \end{array} $}\\
&\resizebox{.83\width}{!}{$
	\begin{array}{cccc|c|ccccc}
	\vect{T} \vect{\alpha}_2 & \vect{t}_{\mu_1+1} + \sum\limits_{r=3}^p\beta_{2,r,1} \vect{t}_{\mu_1+ \dots +\mu_r} & \dots & \vect{t}_{\mu_1 + \mu_2 - 1} + \sum\limits_{r=3}^p\beta_{2,r,\mu_2-1} \vect{t}_{\mu_1+ \dots +\mu_r} &
	\dots &
	\vect{T}\vect{\alpha}_p & \vect{t}_{\mu_1+ \dots + \mu_{p-1} + 1} & \dots & \vect{t}_{n-1}
	\end{array}$}
\bigg ]	= \\
=
\bigg [
&\resizebox{.83\width}{!}{$\begin{array}{cccc|}
	\vect{\check{A}} \vect{t}_1 - \vect{\tilde{\alpha}}_1 & \vect{\check{A}} \vect{t}_2 & \dots & \vect{\check{A}} \vect{t}_{\mu_1}
	\end{array}
$} \\
&\resizebox{.83\width}{!}{$
	\begin{array}{cccc|c|ccccc}
	\vect{\check{A}} \vect{t}_{\mu_1 +1} - \vect{\tilde{\alpha}}_2 & \vect{\check{A}} \vect{t}_{\mu_1 +2} & \dots & \vect{\check{A}} \vect{t}_{\mu_1 + \mu_2} &
	\dots & 
	\vect{\check{A}} \vect{t}_{\mu_1+ \dots + \mu_{p-1} + 1} - \vect{\tilde{\alpha}}_p & \vect{\check{A}} \vect{t}_{\mu_1+ \dots + \mu_{p-1} + 2} & \dots & \vect{\check{A}}\vect{t}_{n}
	\end{array}  $} \bigg ]\text{,}
\end{split}
\end{align}
where $\vect{\alpha}_j = \begin{bmatrix} \alpha_{j,1} & \dots & \alpha_{j,n} \end{bmatrix}\T$, $j = 1,2,\dots,p$.
It follows from \eqref{eq:conditions_on_T_A} that the matrix $\vect{\bar{A}}$ offers the desired structure \eqref{eq:transformed_dynamic_matrix} if
\begin{align}\label{eq:conditions_on_T_A_i}
\begin{split}
\vect{t}_{\mu_1+ \dots +\mu_j - i} + \sum\limits_{r = j+1}^p \beta_{j,r,\mu_j-i} \vect{t}_{\mu_1+ \dots + \mu_r} = \vect{\check{A}} \vect{t}_{\mu_1+ \dots +\mu_j - i +1}
\end{split}
\end{align}
for $j = 1,2,\dots,p$, $i = 1,2,\dots, \mu_{j} - 1$, is satisfied. This is shown by finally substituting $\vect{t}_{\mu_1+ \dots +\mu_j - i}$ and $\vect{t}_{\mu_1+ \dots +\mu_j - i +1}$ into \eqref{eq:conditions_on_T_A_i} according to \eqref{eq:t_i}, i.e.,
\begin{align}\label{eq:conditions_on_T_A_i_2}
\begin{split}
&\vect{\check{A}}^{i}\vect{t}_{\mu_1 + \dots + \mu_j} - \sum\limits_{r = j+1}^{p}\sum\limits_{q=1}^{i} \beta_{j,r,\mu_j - q} \vect{\check{A}}^{i-q}\vect{t}_{\mu_1+ \dots +\mu_r} + \sum\limits_{r = j+1}^p \beta_{j,r,\mu_j-i} \vect{t}_{\mu_1+ \dots + \mu_r}
= \\
=&
\vect{\check{A}} \left(\vect{\check{A}}^{i-1}\vect{t}_{\mu_1 + \dots + \mu_j} - \sum\limits_{r = j+1}^{p}\sum\limits_{q=1}^{i-1} \beta_{j,r,\mu_j - q} \vect{\check{A}}^{i-q-1}\vect{t}_{\mu_1+ \dots +\mu_r}\right)\text{,}%
\end{split}
\end{align}
which is true for all $j = 1,2,\dots,p$, $i = 1,\dots,\mu_j - 1$.

\subsection{Proof of Lemma \ref{LEM:TRANSFORMATION_LEMMA}e)}\label{app1.1b1}

Insertion of the identity matrix $\obsvred^{-1}\obsvred$ into the transformation rule \eqref{eq:matrices_of_transformed_system_D} for the unknown-input matrix $\vect{\bar{D}}$ of the transformed system yields
\begin{align}\label{eq:proof_Dbar_1}
\begin{split}
\vect{\bar{D}} &\overset{\eqref{eq:matrices_of_transformed_system_D}}{=} \vect{T}^{-1}\vect{D} = \vect{T}^{-1}\obsvred^{-1} \obsvred \vect{D}.
\end{split}
\end{align}
Note that $\obsvred$ is invertible which directly follows from Lemma \ref{LEM:TRANSFORMATION_LEMMA}a).
Condition \eqref{eq:relative_degree_condition} allows to express the product $\obsvred \vect{D}$ as
\begin{align}\label{eq:proof_Dbar_3}
\begin{split}
\obsvred\vect{D} &\overset{\eqref{eq:reduced_observability_matrix}}{=}
\begin{bmatrix}
\vect{\check{c}}_1\T \vect{D}\\
\vdots \\
\vect{\check{c}}_1\T \vect{\check{A}}^{\mu_1 - 2} \vect{D}\\
\vect{\check{c}}_1\T \vect{\check{A}}^{\mu_1 - 1} \vect{D}\\
\hline
\vdots \\
\hline
\vect{\check{c}}_p\T \vect{D}\\
\vdots \\
\vect{\check{c}}_p\T \vect{\check{A}}^{\mu_p - 2} \vect{D}\\
\vect{\check{c}}_p\T \vect{\check{A}}^{\mu_p - 1} \vect{D}
\end{bmatrix}
\overset{\eqref{eq:relative_degree_condition}}{=}
\begin{bmatrix}
\vect{0}\T\\
\vdots \\
\vect{0}\T\\
\vect{\check{c}}_1\T \vect{\check{A}}^{\mu_1 - 1} \vect{D}\\
\hline
\vdots \\
\hline
\vect{0}\T \\
\vdots \\
\vect{0}\T\\
\vect{\check{c}}_p\T \vect{\check{A}}^{\mu_p - 1} \vect{D}
\end{bmatrix} =
\sum\limits_{j=1}^{p} \vect{e}_{\mu_1+\dots+\mu_j} \vect{\check{c}}_j\T \vect{\check{A}}^{\mu_j - 1} \vect{D}
.
\end{split}
\end{align}
Substituting \eqref{eq:proof_Dbar_3} into \eqref{eq:proof_Dbar_1} and taking into account \eqref{eq:t_mu_j_major} yields
\begin{align}\label{eq:proof_Dbar_4}
\begin{split}
\vect{\bar{D}} &\overset{\eqref{eq:proof_Dbar_3}}{=} \sum\limits_{j=1}^{p} \vect{T}^{-1}\obsvred^{-1} \vect{e}_{\mu_1+\dots+\mu_j} \vect{\check{c}}_j\T \vect{\check{A}}^{\mu_j - 1} \vect{D} = \\ &\overset{\eqref{eq:t_mu_j_major}}{=} \sum\limits_{j=1}^{p} \vect{T}^{-1} \vect{t}_{\mu_1+\dots+\mu_j} \vect{\check{c}}_j\T \vect{\check{A}}^{\mu_j - 1} \vect{D}.
\end{split}
\end{align}
The vector $\vect{t}_{\mu_1+\dots+\mu_j}$ corresponds to the $(\mu_1+\dots+\mu_j)^\text{th}$ column of $\vect{T}$, i.e., $\vect{t}_{\mu_1+\dots+\mu_j} = \vect{T} \vect{e}_{\mu_1+\dots+\mu_j}$, which allows to simplify \eqref{eq:proof_Dbar_4} to
\begin{align}\label{eq:proof_Dbar_5}
\begin{split}
\vect{\bar{D}} &= \sum\limits_{j=1}^{p} \vect{T}^{-1} \vect{T} \vect{e}_{\mu_1+\dots+\mu_j} \vect{\check{c}}_j\T \vect{\check{A}}^{\mu_j - 1} \vect{D} = \sum\limits_{j=1}^{p} \vect{e}_{\mu_1+\dots+\mu_j} \vect{\check{c}}_j\T \vect{\check{A}}^{\mu_j - 1} \vect{D} =
\begin{bmatrix}
\vect{0}\T\\
\vdots \\
\vect{0}\T\\
\vect{\check{c}}_1\T \vect{\check{A}}^{\mu_1 - 1} \vect{D}\\
\hline
\vdots \\
\hline
\vect{0}\T \\
\vdots \\
\vect{0}\T\\
\vect{\check{c}}_p\T \vect{\check{A}}^{\mu_p - 1} \vect{D}
\end{bmatrix}\text{,}
\end{split}
\end{align}
which finally completes the proof. %

\subsection{Proof of Lemma \ref{LEM:TRANSFORMATION_LEMMA}f)}\label{app1.1c}

The output matrix $\vect{\bar{C}}$ of the transformed system can be expressed by its elements, i.e.,
\begin{align}
\begin{split}
\vect{\bar{C}} &\overset{\eqref{eq:matrices_of_transformed_system_C}}{=} \vect{\Gamma}\vect{C}\vect{T} \underset{\eqref{eq:transformation_matrix}}{\overset{\eqref{eq:auxiliary_system_output_step_1}}{=}}
\begin{bmatrix}
\vect{\check{c}}_1\T \\
\vdots \\
\vect{\check{c}}_p\T
\end{bmatrix}
\begin{bmatrix}
\vect{t}_1 & \dots & \vect{t}_n
\end{bmatrix}
=
\begin{bmatrix}
\vect{\check{c}}_1\T\vect{t}_1 & \dots & \vect{\check{c}}_1\T\vect{t}_n \\
\vdots &  & \vdots \\
\vect{\check{c}}_p\T\vect{t}_1 & \dots & \vect{\check{c}}_p\T\vect{t}_n \\
\end{bmatrix}
\text{,}
\label{eq:proof_Cbar_1}
\end{split}
\end{align}
where the single elements can be described as
\begin{align}
\begin{split}
\vect{\check{c}}_{\bar{j}}\T \vect{t}_{\mu_1+ \dots + \mu_j - i} \overset{\eqref{eq:t_i}}{=} \vect{\check{c}}_{\bar{j}}\T \vect{\check{A}}^{i}\vect{t}_{\mu_1 + \dots + \mu_j} - \sum\limits_{r = j+1}^{p}\sum\limits_{q=1}^{i} \beta_{j,r,\mu_j - q} \vect{\check{c}}_{\bar{j}}\T\vect{\check{A}}^{i-q}\vect{t}_{\mu_1+ \dots +\mu_r}
\label{eq:proof_Cbar_2}
\end{split}
\end{align}
for $\bar{j} = 1,\dots,p$, $j = 1,\dots,p$ and $i = 0,\dots,\mu_j -1$.
In the following, the scalar product $\vect{\check{c}}_{\bar{j}}\T \vect{t}_{\mu_1+ \dots + \mu_j - i}$ given in \eqref{eq:proof_Cbar_2} is considered for the cases (a) $\bar{j} < j$, (b) $\bar{j} = j$ and (c) $\bar{j} > j$:
\begin{enumerate}[label=(\alph*)]
	\item For $\bar{j} < j$, all the terms on the right-hand side of \eqref{eq:proof_Cbar_2} vanish which follows directly from \eqref{eq:linear_combination_t_i_3}, i.e.,
	\begin{align}
	\begin{split}
	\vect{\check{c}}_{\bar{j}}\T \vect{t}_{\mu_1+ \dots + \mu_j - i} =
	\underbrace{\vect{\check{c}}_{\bar{j}}\T \vect{\check{A}}^{i}\vect{t}_{\mu_1 + \dots + \mu_j}}_{\overset{\eqref{eq:linear_combination_t_i_3}}{=} 0} - \sum\limits_{r = j+1}^{p}\sum\limits_{q=1}^{i} \beta_{j,r,\mu_j - q} \underbrace{\vect{\check{c}}_{\bar{j}}\T\vect{\check{A}}^{i-q}\vect{t}_{\mu_1+ \dots +\mu_r}}_{\overset{\eqref{eq:linear_combination_t_i_3}}{=} 0} = 0
	\label{eq:proof_Cbar_a}
	\end{split}
	\end{align}
	for $\bar{j} = 1,\dots,p-1$, $j = \bar{j}+1,\dots,p$ and $i = 0,\dots,\mu_j -1$.
	\item For $\bar{j} = j$, \eqref{eq:proof_Cbar_2} simplifies to
	\begin{align}
	\begin{split}
	\vect{\check{c}}_j\T \vect{t}_{\mu_1+ \dots + \mu_j - i} &= \vect{\check{c}}_j\T \vect{\check{A}}^{i}\vect{t}_{\mu_1 + \dots + \mu_j} - \sum\limits_{r = j+1}^{p}\sum\limits_{q=1}^{i} \beta_{j,r,\mu_j - q} \underbrace{\vect{\check{c}}_j\T\vect{\check{A}}^{i-q}\vect{t}_{\mu_1+ \dots +\mu_r}}_{\overset{\eqref{eq:linear_combination_t_i_3}}{=} 0} = \\
	&= \vect{\check{c}}_j\T \vect{\check{A}}^{i}\vect{t}_{\mu_1 + \dots + \mu_j} \overset{\eqref{eq:linear_combination_t_i_3}}{=}
	\begin{cases}
	1 \qquad & \text{if } i = \mu_j - 1 \\
	0 \qquad & \text{else}
	\end{cases}
	\label{eq:proof_Cbar_b}
	\end{split}
	\end{align}
	for $j = 1,\dots,p$, $i = 0,\dots,\mu_j -1$.
	\item In the case $ \bar{j} > j$, only the coefficients $\beta_{j,r,\mu_j - q}$ whose indices satisfy $\mu_j-q \geq \mu_r$ are nonzero, see \eqref{eq:beta_j_0}, and make a contribution to the sum in \eqref{eq:proof_Cbar_2} which allows to change the upper limit of the summation w.r.t. the index $q$ as
	\begin{align}
	\begin{split}
	\vect{\check{c}}_{\bar{j}}\T \vect{t}_{\mu_1+ \dots + \mu_j - i} = \vect{\check{c}}_{\bar{j}}\T \vect{\check{A}}^{i}\vect{t}_{\mu_1 + \dots + \mu_j} - \sum\limits_{r = j+1}^{p}\sum\limits_{q=1}^{\mu_j - \mu_r} \beta_{j,r,\mu_j - q} \vect{\check{c}}_{\bar{j}}\T\vect{\check{A}}^{i-q}\vect{t}_{\mu_1+ \dots +\mu_r}
	\label{eq:proof_Cbar_c1}
	\end{split}
	\end{align}
	for $\bar{j} = 2,\dots,p$, $j = 1,\dots,\bar{j}-1$, $i = 0,\dots,\mu_j - 1$.
	Since the term $\sum\limits_{r = j+1}^{p}\sum\limits_{q=1}^{\mu_j - \mu_r} \beta_{j,r,\mu_j - q} \vect{\check{c}}_{\bar{j}}\T\vect{\check{A}}^{i-q}\vect{t}_{\mu_1+ \dots +\mu_r}$ reassembles the left-hand side of the ${\big((\bar{j}-j+1)\mu_j - (\mu_{j+1} + \dots + \mu_{\bar{j}}) - i\big)^{\text{th}}}$ equation \eqref{eq:beta_equation_sys}, it equals the corresponding right-hand side of \eqref{eq:beta_equation_sys}, i.e.,
	\begin{align}
	\begin{split}
	\sum\limits_{r = j+1}^{p}\sum\limits_{q=1}^{\mu_j - \mu_r} \beta_{j,r,\mu_j - q} \vect{\check{c}}_{\bar{j}}\T\vect{\check{A}}^{i-q}\vect{t}_{\mu_1+ \dots +\mu_r} = \vect{\check{c}}_{\bar{j}}\T \vect{\check{A}}^{i}\vect{t}_{\mu_1 + \dots + \mu_j}
	\label{eq:proof_Cbar_c2}
	\end{split}
	\end{align}
	for $\bar{j} = 2,\dots,p$, $j = 1,\dots,\bar{j}-1$, $i = 0,\dots,\mu_j - 1$.
	Substituting \eqref{eq:proof_Cbar_c2} into \eqref{eq:proof_Cbar_c1} yields
	\begin{align}
	\begin{split}
	\vect{\check{c}}_{\bar{j}}\T \vect{t}_{\mu_1+ \dots + \mu_j - i} = \vect{\check{c}}_{\bar{j}}\T \vect{\check{A}}^{i}\vect{t}_{\mu_1 + \dots + \mu_j} - \vect{\check{c}}_{\bar{j}}\T \vect{\check{A}}^{i}\vect{t}_{\mu_1 + \dots + \mu_j} = 0
	\label{eq:proof_Cbar_c3}
	\end{split}
	\end{align}
	for $\bar{j} = 2,\dots,p$, $j = 1,\dots,\bar{j}-1$, $i = 0,\dots,\mu_j - 1$.
\end{enumerate}
Putting together the intermediate results \eqref{eq:proof_Cbar_a}, \eqref{eq:proof_Cbar_b} and \eqref{eq:proof_Cbar_c3} yields
\begin{align}
\begin{split}
\vect{\check{c}}_{\bar{j}}\T \vect{t}_{\mu_1+ \dots + \mu_j - i} = 
\begin{cases}
1 \qquad & \text{if } \bar{j} = j \text{ and } i = \mu_j - 1 \\
0 \qquad & \text{else}
\end{cases}
\quad \text{for}
\quad 
\bar{j} = 1,\dots,p\text{,}\quad j = 1,\dots,p\text{,}\quad i = 0,\dots,\mu_j -1.
\label{eq:proof_Cbar_3}
\end{split}
\end{align}
Substituting \eqref{eq:proof_Cbar_3} into \eqref{eq:proof_Cbar_1} finally proves that $\vect{\bar{C}}$ has the desired form \eqref{eq:transformed_output_matrix}. 

\end{document}